\documentclass[12pt]{article}

\usepackage[left=2.54cm,right=2.54cm,top=2.54cm,bottom=2.54cm]{geometry}
\usepackage[doublespacing]{setspace}
\usepackage{enumitem}

\usepackage{amsmath}
\usepackage{amssymb}
\usepackage{bm} % for boldface greek letters
\usepackage{mathtools} % for coloneqq
\usepackage{mathabx} % for widebar

\usepackage{amsthm}
\newtheorem{theorem}{Theorem} 
\newtheorem{assumption}{Assumption} 
\newtheorem{lemma}{Lemma}
\newtheorem{definition}{Definition}

\usepackage{apptools}
\AtAppendix{\numberwithin{lemma}{section}}
\AtAppendix{\numberwithin{equation}{section}}

\usepackage{graphicx}
\usepackage{caption}

\usepackage{tabularx} % for tables
\newcolumntype{Y}{>{\centering\arraybackslash}X} % for centered columns
\usepackage{booktabs} % for toprule / bottomrule

\usepackage{subcaption} % for subcaptions in subfigures

\usepackage{url} % to write url

\usepackage[shortcuts]{extdash} % for line breaks at hyphens

\usepackage{comment}

\usepackage{natbib}
\graphicspath{ {./plots/} }

\newcommand{\Xbf}{\mathbf{X}}
\newcommand{\Vbf}{\mathbf{V}}
\newcommand{\Ybf}{\mathbf{Y}}
\newcommand{\Lbf}{\mathbf{L}}
\newcommand{\Hbf}{\mathbf{H}}
\newcommand{\Zbf}{\mathbf{Z}}
\newcommand{\Ibf}{\mathbf{1}}
\newcommand{\Wbf}{\mathbf{W}}
\newcommand{\Rbb}{\mathbb{R}}
\newcommand{\Pbf}{\mathbf{P}}
\newcommand{\Abf}{\mathbf{A}}
\newcommand{\Qbf}{\mathbf{Q}}
\newcommand{\Ubf}{\mathbf{U}}
\newcommand{\Omc}{\mathcal{O}}
\newcommand{\Mmc}{\mathcal{M}}
\newcommand{\XZbar}{\widebar{\Xbf\Zbf}}
\newcommand{\Hbar}{\widebar{H}}
\newcommand{\Yoe}{\Ybf^{0,\varepsilon}}
\newcommand{\Yhoe}{\hat{\Ybf}^{0,\varepsilon}}

\newcommand{\card}[1]{\left\lvert #1 \right\rvert}
\newcommand{\norm}[1]{\left\lVert #1 \right\rVert}
\newcommand{\indep}{\perp \!\!\! \perp}
\newcommand{\floor}[1]{\left\lfloor #1 \right\rfloor}

\DeclareMathOperator*{\argmin}{arg\,min}

\usepackage{xcolor}
\usepackage{ifthen}
\newboolean{showcomments}
\setboolean{showcomments}{true}
\ifthenelse{\boolean{showcomments}}
{ \newcommand{\mynote}[3]{\quad
    \fbox{\bfseries\sffamily\scriptsize#1}{\small$\blacktriangleright$\textsf{\textit{\color{#3}{#2}}}$\blacktriangleleft$\quad}}}
{ \newcommand{\mynote}[3]{}}
\begin{document}
% Reduce the vertical space around equations
\setlength{\abovedisplayskip}{3pt}
\setlength{\belowdisplayskip}{3pt}

\begin{center}
\textbf{\Large Data-driven model selection within the matrix completion method for causal panel data models.} \\

\vspace{1.5cc}
{ Sandro Heiniger$^{1}$}\\

\vspace{0.3 cm}

{\small $^{1}$University of St.Gallen, sandro.heiniger@unisg.ch}
 \end{center}
\vspace{1.5cc}

\begin{abstract}
    \noindent Matrix completion estimators are employed in causal panel data models to regulate the rank of the underlying factor model using nuclear norm minimization. This convex optimization problem enables concurrent regularization of a potentially high-dimensional set of covariates to shrink the model size. 
    %We propose a two-step procedure: first, selecting the optimal regularization penalty parameters by cross-validation, and second, conducting estimation of the model parameters.
    For valid finite sample inference, we adopt a permutation-based approach and prove its validity for any treatment assignment mechanism. Simulations illustrate the consistency of the proposed estimator in parameter estimation and variable selection. An application to public health policies in Germany demonstrates the data-driven model selection feature on empirical data and finds no effect of travel restrictions on the containment of severe Covid-19 infections.

\vspace{0.95cc}
\parbox{24cc}{{\it Matrix completion; Model selection; Panel data; Regularisation;}
}
\end{abstract}

%We propose a two-step procedure: first, selecting the optimal model, and second, conducting estimation without covariate regularization to obtain unbiased estimates of average treatment effects on the treated and model parameters. 

\section{Introduction} \label{Introduction}
The prevalence of big data and high-dimensional covariate spaces has grown significantly in applied economics (see~\cite{jin2015significance} for an overview). Consequently, the econometric literature strives to develop suitable tools to cater to the needs of applied research (e.g., \cite{belloni2014inference, farrell2015robust, belloni2017program, ning2020robust}). One widely employed framework in this context is the $l_1$ regularization, known as \textit{lasso}, which facilitates model selection and feature extraction by regulating covariate parameters. \cite{fan2006statistical} provide an extensive discussion on the pivotal role of model selection for knowledge discovery and addressing high-dimensionality challenges.\par

Rigorous sparse model selection methods offer key advantages in terms of prediction accuracy and interpretability \citep{tibshirani1996regression}. Beyond prediction tasks, lasso techniques have found substantial utility in causal analysis. Even though for causal identification strategies that rely on the unconfoundedness assumption~\citep{rosenbaum1983central}, the role of covariates narrows down to including all confounding variables in the data, model selection nevertheless offers significant features:\par

1) Common estimators often rely on estimating an outcome/propensity score model to evaluate unobserved potential outcomes or as a nuisance function. Precise estimation of those underlying models can lead to significant gains in prediction accuracy for causal estimands in finite samples.\par

2) In the absence of theoretical guidance on the true model, researchers may face a high-dimensional problem, where the number of parameters exceeds available observations (often referred to as '$p \gg n$' or 'The Curse of Dimensionality'). In such scenarios, estimators like the linear regression family are inapplicable due to insufficient degrees of freedom, necessitating dimensionality reduction.\par

3) The econometric literature has proposed various variants of post-regularization estimation methods~\citep{chernozhukov2015valid}. If the initial model selection step achieves sufficient sparsity, a second non-regularized estimation of the resulting parsimonious model may even become the oracle estimator when the true model is identified~\citep{belloni2013least}. \par

4) Recent years have witnessed a significant surge in interest in estimating and analyzing heterogeneous treatment effects \citep{kunzel2019metalearners}. %The question of whether the choice of heterogeneity dimensions can be data-supported remains open. However, a
A model selection feature can offer valuable support and guidance for selecting potential heterogeneity dimensions.\par

All the aforementioned points are strongly supported by the prevalence of low-rank matrices in high-dimensional data, as demonstrated in~\cite{udell2019big}. It has been shown that \textit{lasso} can consistently identify such sparse models while achieving convergence at the optimal rate concerning the parameters of interest (known as the \textit{oracle property}) under specific conditions~\citep{meinshausen2006variable,zou2006adaptive,meinshausen2007relaxed}.\par

%While $l_1$ regularization was initially introduced for linear and generalized regression models \citep{tibshirani1996regression}, it has since been integrated into hazards models~\citep{tibshirani1997lasso}, GMM estimators~\citep{caner2009lasso}, and IV estimators~\citep{belloni2012sparse}. In this work, we extend the use of $l_1$ regularization to matrix completion estimators for causal panel data models.\par

Our main contribution lies in extending the use of $l_1$ regularization to matrix completion estimators for causal panel data models. In their seminal work, \cite{athey2021matrix} lay the foundation for employing matrix completion methods in causal panel data models. This estimator utilizes observed elements of control outcomes to impute the unobserved (\textit{missing}) potential outcomes of treated observations in the untreated state. By interpreting the panel data structure as a matrix, this procedure approximates the complete matrix of control state potential outcomes by regularizing the complexity of the underlying factor model using the nuclear norm. Matrix completion estimation has been shown to outperform synthetic control estimators or regression methods based on the unconfoundedness assumption~\citep{athey2019ensemble, athey2021matrix}, and its initial applications have shown its practical value~\citep{wood2020procedural, rafaty2020carbon, levy2022effects}. Appendix~\ref{app:Matrix completion estimator} provides background of the matrix completion estimator and an embedding in the context of causal panel data models.\par

In this work, we leverage the convex nature of nuclear norm minimization to regularize the rank of the general factor matrix. In a model with covariates, this allows for concurrent regularization of a potentially high-dimensional set of variables. Unlike other estimation methods, where $l_1$ regularization often increases the numerical complexity of an estimator by lacking a closed-form solution, preserves $l_1$ regularization of covariate parameters the convex optimization of rank regularization in matrix completion estimators, thus adding no additional complexity. The addition of $l_1$ regularization to the covariate space enhances the matrix completion framework with a model selection property, thus enhancing its performance in the context of causal panel data analysis without increasing the computational complexity. The methodology allows for a two-step approach similar to common post-regularisation estimators: Involving the selection of the optimal model as the first step and conducting a second estimation without covariate regularization ensuring unbiased estimates of model parameters. \par

We adopt the permutation-based inference procedure presented in \cite{chernozhukov2021exact} to test the sharp null hypothesis of a zero treatment effect, accommodating a stationary and weakly dependent shock process. The validity of this inference procedure is established not only for the proposed estimator with integrated $l_1$ covariate regularization but also, as a second contribution, for settings with any treatment assignment scheme. This contribution extends the set of potential applications of matrix completion methods with integrated model selection properties by enabling their use in diverse settings with various treatment assignment mechanisms. This advancement enhances the robustness and versatility of our proposed approach, making it well-suited for a wide range of real-world scenarios.\par

% Preview of simulation results
A numerical simulation using the proposed estimator demonstrates that the introduced $l_1$ regularization on the covariate space within the matrix completion estimator effectively reduces the determined model size, as intended. Employing a cross-validation step to determine the optimal penalty parameter value justifies adopting a \textit{1se} optimality criterion to obtain the ideal model size. The resulting model sizes using the \textit{1se} condition are shown to be precise and appear to converge to the true values of the underlying model for large sample sizes.\par.

The proposed estimator outperforms the matrix completion estimator without covariate regularization, in particular for small sample sizes or strong signal-to-noise ratios, and it remains equally precise in other settings. The permutation-based approach necessitates enforcing the null hypothesis for exact and valid inference. However, this requirement may lead to a biased treatment effect estimation if the null hypothesis does not hold. In Section \ref{sec:Data-driven model selection}, we shed light on how this impacts the estimation procedure and identify the conditions under which a simple rule-of-thumb correction of the treatment effect estimate can be applied. The simulation results reveal that this correction partially mitigates the downward bias and achieves equal estimation accuracy as the unbiased estimates when discarding the null hypothesis required for inference. Executing a two-stage approach with a second non-regularized estimation on the initially determined model does not yield an improved estimation of the treatment effect on the treated.\par

%Application
An illustrative application examining the impact of travel restriction policies implemented during the Covid-19 pandemic on the incidence of infections requiring treatment in intensive care units underscores the advantageous properties of the model selection mechanism. Within a high-dimensional context characterized by a presumed minuscule signal-to-noise ratio, the proposed estimator adeptly identifies a notably sparse model. The empirical application using panel data aligns precisely with the anticipated behavior of the estimator, as derived from the simulation study. The findings reveal that the estimated effect of the mandatory testing requirement upon entry from foreign countries with heightened incidence rates is small coming along with a notably large p-value. This compelling evidence suggests that the influence of this particular travel restriction policy on public health outcomes is negligible.\par

% R package
To facilitate the adoption and replication of our methodology, we have developed an R-package that implements our proposed estimator. The package is publicly available for download on our GitHub repository.\footnote{\url{github.com/heinigersandro/MC_MS}} We encourage researchers to utilize our methodology, validate our findings, and contribute to advancements in the field.\par

%remainder of the paper
The remainder of the article reads as follows: In Section \ref{sec:Data-driven model selection}, we explain how the model selection property is incorporated into the estimator. Section \ref{sec:Theoretical properties} discusses the theoretical properties of the proposed estimator. The behavior of the proposed estimator in simulations is presented in Section \ref{sec:Simulations}. Section \ref{sec:Illustration} discusses the results of the illustrative application to Covid-19 infections in Germany. Finally, we conclude in Section \ref{sec:Conclusion}. In the Appendix, we provide an introduction to the matrix completion estimator for causal panel data models, additional estimation results and proofs for the theoretical findings.

\section{Data-driven model selection within matrix completion estimators}\label{sec:Data-driven model selection}

We consider a panel of $N$ units observed over $T$ periods, where each unit $i$ potentially undergoes a binary treatment in time-period $t$, $W_{it} \in \{0,1\}$, with corresponding potential outcomes $Y_{it}(0) \coloneqq Y_{it}(W_{it}=0)$ and $Y_{it}(1) \coloneqq Y_{it}(W_{it}=1)$. Our object of interest is the average treatment effect on the treated\footnote{As reasoned in Appendix~\ref{app:Matrix completion estimator}, we assume most observations to be in the untreated state and therefore restrict the focus on the ATET.}, defined as follows:
\begin{equation} \label{eq:tau_exp}
\tau_{\text{ATET}}=\mathbb{E}_{(i,t):W_{it}=1}\left[Y_{it}(1)-Y_{it}(0)\right].
\end{equation}

We leverage the model with covariates proposed in \cite{athey2021matrix}:
\begin{equation}\label{eq:Full model}
\Ybf=\bm{\Theta}\circ \Wbf + \Lbf^*+\Xbf\Hbf^*\Zbf + \left[\Vbf_{it}^\top\bm{\beta}^*\right]_{it} + \mathbf{\Gamma}^*\Ibf_T^\top + \Ibf_N(\mathbf{\Delta}^*)^\top + \Ubf,
\end{equation}
where $\Ybf \!\in\!\Rbb^{N\!\times\!T}$ represents the realized outcomes, $\bm{\Theta} \!\in\!\Rbb^{N\!\times\!T}$ is the matrix of potentially heterogeneous treatment effects, $\Wbf \!\in\!\Rbb^{N\!\times\!T}$ denotes the treatment allocation, $\Lbf^* \!\in\!\Rbb^{N\!\times\!T}$ represents the unobserved factor matrix of low-rank, $\Xbf \!\in\!\Rbb^{N\!\times\!P}$ contains unit-specific covariates, $\Zbf\!\in\!\Rbb^{Q\!\times\!T}$ contains time-specific covariates, $\Vbf_{it}\!\in\!\Rbb^{J}$ contains the unit-time varying covariates for unit $i$ at time $t$, $\Hbf^* \!\in\!\Rbb^{P\!\times\!Q}$ and $\bm{\beta}^* \!\in\!\Rbb^{J}$ are the unknown covariate parameters, $\bm{\Gamma}\!\in\!\Rbb^{N}$ represents unit-level fixed effects, $\bm{\Delta}\!\in\!\Rbb^{T}$ represents time-level fixed effects, and $\Ubf \!\in\!\Rbb^{N\!\times\!T}$ is a random shock term.\par

\begin{assumption}[Identifying conditions]\label{ass:identifying condictions}
Assume that $E[\Ubf \:\vert\: \Lbf, \Hbf, \bm{\beta}, \bm{\Gamma}, \bm{\Delta}]=0$, and that $\forall_{i\neq j}: \Ubf_{i,\cdot} \indep \Ubf_{j,\cdot} \; \vert\; \Lbf, \Hbf, \bm{\beta}, \bm{\Gamma}, \bm{\Delta}$.
\end{assumption}

\begin{assumption}[Regularity of stochastic shock]\label{ass:regularity of stochastic shock}
Suppose that the observed outcomes follow the model \eqref{eq:Full model}. Assume that for each observation $i$, the stochastic process $\{\Ubf_{it}\}_{t=1}^T$ satisfies one of the following conditions:
\begin{enumerate}[itemsep=-6pt]
    \item $\{\Ubf_{it}\}_{t=1}^T$ are iid, or
    \item $\{\Ubf_{it}\}_{t=1}^T$ are stationary, strongly mixing, with sum of mixing coefficients bounded by M.
\end{enumerate}
\end{assumption}

Following the literature on factor models, we use an additive and linear model in \eqref{eq:Full model}. Despite appearing restrictive, the set of covariates is not limited in size, allowing for the inclusion of any functional forms if desired. More complex interactive terms are absorbed by the latent matrix $\Lbf^*$ and potentially by the fixed effects as well, granting the model \eqref{eq:Full model} sufficient flexibility to cover a broad range of applications.\par

The structure of $\Hbf$ in the outcome model \eqref{eq:Full model} allows only for linked unit- and time-covariate effects. To include linear terms in both $\Xbf$ and $\Zbf$, a more comprehensive model can be defined as follows:
\begin{equation}\label{eq:Full model linear}
\Ybf=\bm{\Theta}\circ\Wbf + \Lbf^*+\tilde{\Xbf}\tilde{\Hbf}^*\tilde{\Zbf} + \left[\Vbf_{it}^\top\bm{\beta}^*\right]_{it} + \mathbf{\Gamma}^*\Ibf_T^\top + \Ibf_N(\mathbf{\Delta}^*)^\top + \Ubf,
\end{equation}
where $\tilde{\Xbf}=[\Xbf|\mathbf{I}_{N\!\times\!N}]$ and $\tilde{\Zbf}=[\Zbf^\top|\mathbf{I}_{T\!\times\!T}]^\top$. For brevity, we focus on the covariate setting as in \eqref{eq:Full model} since all adaptations to the richer model \eqref{eq:Full model linear} are straightforward.\par

We define the set of treated observations as $\mathcal{M}\coloneqq \left\{(i,t) \text{ with } W_{it}=1\right\}$ and the set of control observations as $\mathcal{O}\coloneqq \left\{(i,t) \text{ with } W_{it}=0\right\}$. Correspondingly, we define the matrices projecting on the respective treatment allocation space as follows:
 \begin{align*}
    \Pbf_{\Omc}(A)_{it} &= \begin{cases} \Abf_{it} & \text{if } (i,t) \in \Omc \\[-10pt]
    0 & \text{if } (i,t) \notin \Omc 
    \end{cases}\\
    \Pbf_{\Mmc}(A)_{it} &= \begin{cases} \Abf_{it} & \text{if } (i,t) \in \Mmc \\[-10pt]
    0 & \text{if } (i,t) \notin \Mmc 
    \end{cases}.
\end{align*}
Using the defined notation and supposing Assumption \ref{ass:identifying condictions} holds, we can rewrite the sample representation of the ATET, which is our object of interest, as follows:
\begin{equation} \label{eq:tau_sample} \hat{\tau}_{\text{ATET}}=\frac{1}{\card{\Mmc}}\sum_{(i,t) \in \Mmc} \Pbf_{\Mmc}(\Ybf-\hat{\Ybf}(0)).\end{equation}

To evaluate the estimator~\eqref{eq:tau_sample}, we only need the matrix of potential outcomes in the absence of treatment, denoted by $\Ybf(0)$. As for each observation, only one potential outcome is realized, some elements in $\Ybf(0)$ are missing as they are not observed in practice. The outcome model~\eqref{eq:Full model} directly defines the potential outcome model as follows:
\begin{equation}\label{eq:PO model}
\Ybf(0)=\Lbf^*+\Xbf\Hbf^*\Zbf + \left[\Vbf_{it}^\top\bm{\beta}^*\right]_{it} + \mathbf{\Gamma}^*\Ibf_T^\top + \Ibf_N(\mathbf{\Delta}^*)^\top + \Ubf.
\end{equation}

We estimate the unknown parameters of the potential outcome model \eqref{eq:PO model} using a matrix completion method similar to~\cite{athey2021matrix} and~\cite{chernozhukov2021exact}, but with regularization on the full covariate space\footnote{In \cite{athey2021matrix}, the model with covariates drafts a penalty term on the link matrix $\Hbf$ between unit- and time-varying covariates, but not on the unit-time-varying covariate parameters $\bm{\beta}$. Though, the authors do not discuss the implications of covariate space regularization on the model selection properties of the proposed estimator.}\footnote{In principle, the unit- and time-level fixed effects could be regularized as well. We follow the arguments of~\cite{hastie2009elements} that regularizing the intercept in $l_1$ regularization increases the bias (see also the discussion in~\cite{athey2021matrix}).}. We adopt the permutation-based procedure for valid finite sample inference described in \cite{chernozhukov2021exact}. The method requires the enforcement of the null hypothesis of a zero treatment effect to obtain a valid and exact inference procedure.\footnote{The intuition for this is that the permutation-based inference requires the estimation to be provably accurate across the whole unit-time space to obtain equally distributed residuals. Enforcing the null allows us to use all information in the data, regardless of the treatment state.} In fact, \cite{chernozhukov2021exact} show that the rejection rates are substantially erroneous if the null is not enforced. The presumed absence of a treatment effect implies that there is no missing data, which is conterminous to estimating the model parameters using all observations in the panel. This fundamentally differentiates the \cite{chernozhukov2021exact} estimation approach from \cite{athey2021matrix} which guarantees finite sample bounds for the effect estimates but no valid inference.\par

The estimation procedure is as follows:
\begin{equation}\label{eq:Estimator}\begin{split}
(\hat{\Lbf},\hat{\Hbf}, \hat{\bm{\beta}}, \hat{\bm{\Gamma}},\hat{\bm{\Delta}})=\argmin_{\Lbf,\Hbf, \bm{\beta},\bm{\Gamma},\bm{\Delta}}\left\{ \frac{1}{NT} \norm{\Ybf - \Lbf- \Xbf\Hbf\Zbf - \left[\Vbf_{it}^\top\bm{\beta}\right]_{it} - \mathbf{\Gamma}\Ibf_T^\top - \Ibf_N(\mathbf{\Delta})^\top}_F^2 \right. \\ \left. \vphantom{\frac{1}{\card{\Omc}}} + \lambda_L\norm{\Lbf}_* + \lambda_H\norm{\Hbf}_{1,e}+\lambda_{\beta} \norm{\bm{\beta}}_{1,e} \right\},
\end{split}\end{equation}
where $\norm{\Abf}_F^2=\sum_{(i,t)} \Abf_{it}^2$ is the squared Fröbenius norm, $\norm{\Abf}_*=\sum_{i=1}^N \sigma_i(\Abf)$ is the nuclear norm, and $\norm{\Abf}_{1,e}=\sum_{(i,t)} \card{\Abf_{it}}$ is the element-wise $l_1$ norm. \par

The choice of the nuclear norm on $\Lbf$ is crucial as it regularizes the rank of the matrix through its singular values $\sigma_i(\Lbf)$. Using other norms, such as the Frobenius or element-wise $l_1$ norm, would not be suitable, as the objective function would be minimized by setting $\Lbf_{it}=0$ for all $(i,t) \in \Mmc$ due to the projection on the control space $\Pbf_\Omc$. The rank norm $\norm{\Abf}_0=\sum_{i=1}^N \Ibf\{\sigma_i(\Lbf)>0\}$,, which might be a preferred choice, is computationally infeasible as the rank norm optimization is NP-hard~\citep{recht2010guaranteed}. On the other hand, the choice of the other norms is straightforward. The Frobenius norm on the prediction error represents an MSE-like measure, while the $l_1$ regularization on the covariate parameters facilitates the model selection process.\par

The proposed estimator \eqref{eq:Estimator} can be efficiently computed by iteratively applying a \textit{soft-impute} step~\citep{mazumder2010spectral} to handle the nuclear norm rank regularization of the matrix of unobserved factors, and a gradient descent step~\citep{friedman2010regularization} with respect to the $l_1$ regularization of the covariate space. It is important to note that both steps remain within the realm of convex optimization, making the computation computationally efficient.\par 

The values of the regularization parameters $\lambda_{L,H,\beta}$ are selected through cross-validation. To ensure that the estimation of the ATET produces precise potential outcomes under no treatment, we aim to find the $\lambda_{L,H,\beta}$ values that minimize the prediction error for potential outcomes. To maintain independence from whether the null hypothesis holds or not, we restrict the cross-validation sample to the observations in the control state. For a specified number of folds $K$, we randomly draw $K$ sets $\Omc_k\subset \Omc$ with a size $\frac{\card{\Omc_k}}{\card{\Omc}}=\frac{\card{\Omc}}{NT}$, corresponding to the overall share of untreated observations. The evaluation set for each fold is given by $\Omc_{k}^{\boldsymbol{-}} \coloneqq (i,t):(i,t)\!\in \!\Omc , (i,t)\! \notin\! \Omc_k$. For a chosen set $\Lambda$, we search for the optimal $\lambda$-configuration by minimizing the mean squared prediction error over all partitions as follows:
\begin{align*}\begin{split}(\lambda_L,\lambda_H,\lambda_{\beta})=\argmin_{\left(\lambda_{L'},\lambda_{H'}, \lambda_{\beta'}\right)\in \Lambda}\sum_{k=1}^K \left\lvert\left\lvert\Pbf_{\Omc_k^{\boldsymbol{-}}} \left(\Ybf - \hat{\Lbf'}_k- \Xbf\hat{\Hbf'}_k\Zbf - \left[\Vbf_{it}^\top\hat{\bm{\beta'}}_k\right]_{it} \right.\right.\right. \\ \left.\left.\left.- \hat{\mathbf{\Gamma'}}_k\Ibf_T^\top - \Ibf_N(\hat{\mathbf{\Delta'}}_k)^\top\right)\vphantom{\Pbf_{\Omc_k^{\boldsymbol{-}}}}\right\rvert\right\rvert_F^2.\end{split}
\end{align*}
Here, $(\hat{\Lbf'}_k,\hat{\Hbf'}_k, \hat{\bm{\beta'}}_k, \hat{\bm{\Gamma'}}_k,\hat{\bm{\Delta'}}_k)$ are estimated by \eqref{eq:Estimator} on each training set $\Omc_k$ using the penalty parameters $\left(\lambda_{L'},\lambda_{H'}, \lambda_{\beta'}\right)$.
\par

To restrict the set of to-be-evaluate elements $\Lambda$, we identify the theoretical minimal values of $\lambda_{L,H,\beta}$ such that all parameter estimates in the respective objects are reduced to 0. Along with $\lambda_{L,H,\beta}=0$, this establishes two natural bounds for each penalization parameter. The optimization search can be conducted on a standard three-dimensional grid. The MCMS R-package utilizes a more advanced hyper-cube search that leverages the convex optimization to expedite convergence to the minimizing configuration. This approach enhances the efficiency of the parameter selection process.\par

A natural extension of estimators with integrated model selection is to implement a two-stage procedure. In the first stage, the estimation procedure is run as usual to identify the sparse model, and in the second stage, the estimation is performed without regularization on the identified subset of covariates~\citep{chernozhukov2015valid}. The objective of the second stage is to obtain an unbiased estimate of model parameters by dropping the regularization. For the matrix completion method, the analogous two-stage estimator first identifies the informative covariates and determines the rank of $\Lbf$ in the first step. Then, in the second step, it applies an unregularized estimation under the restrictions on the covariate space and the rank of the factor model. \par

%Note that the proposed version of the estimator does use the information on realized outcomes in the treatment state. In most applications, where the number of treated units is small, this loss of information is minimal. Nonetheless, in case the individual treatment effects follow a known pattern, natural extensions of the estimator to include $\Ybf(1)$ in the objective function are possible.\par

The estimation of the model \eqref{eq:Estimator} is obviously biased if the null hypothesis is violated. In such cases, the treatment effect might be either evenly dispersed over all observations and periods or allocated to confounding variables, leading to a reduction in the ATET estimate. If the assessment of p-values is not a need, an estimation of the model parameters using only the observations in the absence of treatment yields more precise estimates of the treatment effect:
\begin{equation}\label{eq:Estimator_Athey}\begin{split}
(\hat{\Lbf},\hat{\Hbf}, \hat{\bm{\beta}}, \hat{\bm{\Gamma}},\hat{\bm{\Delta}})=\argmin_{\Lbf,\Hbf, \bm{\beta},\bm{\Gamma},\bm{\Delta}}\left\{ \frac{1}{\card{\Omc}} \norm{\Pbf_\Omc \left(\Ybf - \Lbf- \Xbf\Hbf\Zbf - \left[\Vbf_{it}^\top\bm{\beta}\right]_{it} - \mathbf{\Gamma}\Ibf_T^\top - \Ibf_N(\mathbf{\Delta})^\top\right)}_F^2 \right. \\ \left. \vphantom{\frac{1}{\card{\Omc}}} + \lambda_L\norm{\Lbf}_* + \lambda_H\norm{\Hbf}_{1,e}+\lambda_{\beta} \norm{\bm{\beta}}_{1,e} \right\},
\end{split}\end{equation}\par

It is important to note that if the treatment effect is independent of the covariates or the unconfoundedness assumption holds, the treatment effect is distributed evenly across all observations. Under the assumption of an independent treatment effect, each observation receives a share of the average effect among the treated sample, given by $\frac{\card{\Mmc}}{\card{\Omc}}$. If we assume homogeneity, this share is equal to the ATET. However, in situations with strongly heterogeneous treatment effects, the average effect in the sample may diverge from the ATET. Additionally, the treatment allocation mechanism also affects the dispersion of treatment effects among the observations under the unconfoundedness assumption. The split of the sample average effect assigns more weight to observations with a higher propensity score, as the treatment effect is partly allocated to confounding variables. To correct the ATET for marginal treatment effect and propensity score heterogeneity, a simple rule-of-thumb correction is given by: \begin{equation}
    \hat{\tau}_\text{ATET,rot\_cor}=\frac{NT}{\card{\Omc}}\;\hat{\tau}_\text{ATET} \label{eq:tau_correction}
\end{equation}    

\section{Theoretical properties}\label{sec:Theoretical properties}

\begin{definition}[Definition of test statistic $S$] \label{def:test statistic} The test statistic $S(\hat{U})$ is defined as follows: $$ S(\hat{U}) = \card{\Mmc}^{-1}\;\sum\nolimits_{(i,t)\in\Mmc} \card{\hat{U}_{it}} $$
\end{definition}
While other test statistics based on element-wise p-norms are theoretically possible, the test statistic in Definition \ref{def:test statistic} has been found to exhibit good properties when estimating average treatment effects \citep{chernozhukov2021exact}. 
\begin{assumption}[Bounded test statistic]\label{ass:test_statistic}
Assume that, under the null hypothesis, the density function of the test statistic $S(U)$ exists and is bounded.   
\end{assumption}\par

Let the actual and estimated shock-discarded potential outcomes in the absence of treatment be denoted as:
\begin{align*}
    \Yoe & \coloneqq \Lbf^*+\Xbf\Hbf^*\Zbf + \left[\Vbf_{it}^\top\bm{\beta}^*\right]_{it} + \mathbf{\Gamma}^*\Ibf_T^\top + \Ibf_N(\mathbf{\Delta}^*)^\top=\Ybf(0)-\Ubf \\
    \Yhoe & \coloneqq \hat{\Lbf}+\Xbf\hat{\Hbf}\Zbf + \left[\Vbf_{it}^\top\hat{\bm{\beta}}\right]_{it} + \hat{\mathbf{\Gamma}}\Ibf_T^\top + \Ibf_N(\hat{\mathbf{\Delta}})^\top=\hat{\Ybf}(0)
\end{align*}
    
\begin{assumption}[Consistency of the counterfactual estimation]\label{ass:consistency}
    Suppose that $\Yhoe$ is a mean-unbiased predictor of $\Yoe$. Let there be two sequences $\gamma_{N,T},\delta_{N,T}$ converging to 0 in $NT$. Assume with probability $1-\gamma_{N,T}$ that 
 \begin{enumerate}[itemsep=-4pt]
    \item $(NT)^{-1}\norm{\Yhoe-\Yoe}_2^2\leq \delta_{N,T}^2$ (small estimation mse), and
    \item $\forall_{(i,t)\in\Omc}: \card{\Yhoe-\Yoe}\leq \delta_{N,T}$ (small pointwise estimation error).
\end{enumerate}   
\end{assumption}
Assumption \ref{ass:consistency} is easy to verify for various types of estimators, including the basic non-regularized version of the matrix completion estimator, as shown in \cite{chernozhukov2021exact}. The following lemma elaborates on the primitive conditions to ensure that the proposed estimator with covariate regularization satisfies Assumption \ref{ass:consistency}.
\begin{lemma}\label{lem:consistency}
Consider the Estimator \eqref{eq:Estimator} and assume the conditions stated at the beginning of the proof, then $(NT)^{-1}\norm{\Yhoe-\Yoe}_2^2 = o_P(1)$ and $\forall_{(i,t)\in \Mmc}\: : \: \card{\Yhoe_{it}-\Yoe_{it}}= o_P(1)$.
\end{lemma}\par
We make use of permutations to compute the p-values. 
\begin{definition}[Set of permutations] \label{def:permutations}
If Assumption \ref{ass:regularity of stochastic shock}.1 holds, a permutation is a one-to-one mapping $\pi : \{1,\dots,N\}\times\{1,\dots,T\} \mapsto \{1,\dots,N\}\times\{1,\dots,T\}$. In case of Assumption \ref{ass:regularity of stochastic shock}.2, a permutation is the set of horizontal moving block permutations for each individual $\pi=\{\pi_i^{t'}\}_{i=1}^N$ where the respecting observation-internal moving block permutation is defined as  $$\pi_{t'}(t)=\floor{t+t'}_T$$ 
The set of all possible permutations is denoted by $\Pi$. For each $\pi \in \Pi$, we define the matrix of permuted residuals $\hat{U}_\pi \coloneqq \left[\hat{U}_{\pi(i,t)}\right]_{i,t} $, or $\hat{U}_\pi \coloneqq \left[\hat{U}_{i,\pi(t)}\right]_{i,t} $ respectively.
%We make use of permutations to compute the p-values. A permutation is a one-to-one mapping $\pi : \{1,\dots,T\} \mapsto \{1,\dots,T\}$ and applies equally to all observations. The set of all possible permutations is denoted by $\Pi$. For each $\pi \in \Pi$, we define the matrix of permuted residuals $\hat{U}_\pi \coloneqq \left[\hat{U}_{i,\pi(t)}\right]_{i,t}$.
\end{definition}
For the finite sample validity of the matrix completion estimator, the choice of $\Pi$ is irrelevant. However it is worth to note that set of one-to-one mappings on $ \{1,\dots,N\}\times\{1,\dots,T\}$ is larger than the set of horizontal moving block permutations allowing for the computation of more precise p-values under Assumption \ref{ass:regularity of stochastic shock}.1. 

\begin{definition}[Definition of p-value]
    The p-value is defined as $$\hat{p}=1- \hat{F}\left(S(\hat{U})\right), \;\text{where}\; \hat{F}(x)=\frac{1}{\card{\Pi}}\sum_{\pi \in \Pi} \Ibf\left\{S\left(\hat{U}_\pi\right)<x\right\}$$
\end{definition}

\begin{theorem}[Approximate validity of p-value] \label{th:pvalue}
Suppose that Assumptions \ref{ass:regularity of stochastic shock}-\ref{ass:consistency} hold. Then under the null hypothesis, the p-value is approximately unbiased in size: $$ \card{P[\hat{p}\leq\alpha]-\alpha} \leq C(\tilde{\delta}_{N,T} + \delta_{N,T}+\sqrt{\delta_{N,T}}+\gamma_{N,T})$$ where $\tilde{\delta}_{N,T}=(\card{\Mmc}/\card{\Omc})^{1/4}(\log(NT))$ and the constant $C$ depends on $\card{\Mmc}$ but not on $N$ and $T$.
\end{theorem}

Note that the results of Theorem \ref{th:pvalue} are non-asymptotic, i.e. they hold in finite sample. The finite sample bounds of the size properties imply that the inference procedure is exact in $NT\rightarrow\infty$. Naturally, the power of the hypothesis testing depends on characteristics of the underlying model \eqref{eq:Full model}. For instance, if the model is estimated more accurately or if the shocks have a smaller variance, the hypothesis tests tend to have higher power.

\section{Simulation}\label{sec:Simulations}
In the following section, the results are only presented for matrix $\Hbf$ that links the unit-specific and the time-specific covariates. All findings equivalently apply for $\bm{\beta}$, the parameter of the unit-time-varying covariates. The corresponding figures are presented in Appendix~\ref{app:Simulation Results beta}.

\subsection{Data-generating process}
For the subsequent simulations, we use the following data-generating process: The outcome matrix follows the linear model in~\eqref{eq:Full model} with a homogeneous treatment effect\footnote{The matrix completion estimator can naturally deal with heterogeneous treatment effects. As our object of interest is the average treatment effect on the treated, we stay with homogeneous effects for the simulation.}, interactive unit- and time-specific covariates, unit-time varying covariates, unit fixed-effects and time fixed-effect terms. The treatment allocation $\mathbf{W}$ follows a Bernoulli model. The latent factor model matrix $\Lbf$ is random matrix restricted to a specified $\text{rank}_L$ with exponentially distributed singular values. The unit-specific covariates of size $p$ are multi-variate normal with an idiosyncratic unit factor, time-specific covariates of size $q$ are generated accordingly. The coefficients of $\Hbf$, the link between those covariates, are normally distributed, but only a random fraction is \textit{active} while the other elements are set to 0. The data-generating process omits linear terms in the unit- or time-specific covariates as discussed in Section~\ref{sec:Data-driven model selection}. The unit-time varying covariates $\Vbf_{it}$ are independent standard normal and of the true coefficients $\bm{\beta}$, again, only a fraction is active. \par
The exact formulations of all components and the default parameter values are described in Appendix~\ref{app:Formulation of data generating process}.

\subsection{Choice of optimal penalty parameters by cross-validation}

Figure~\ref{fig:fan_H} shows the paths of the estimated coefficients in $\hat{\Hbf}$ for different values of $\lambda_{H}$. We observe the known pattern of $l_1$ regularisation that coefficients with small attributed parameter values are wiped out for already low penalty parameters. With increasing regularization, one variable by another is eliminated until even the variables with largest absolute coefficient estimates are regularized out. The descriptive illustrations on the size and the number of coefficient parameters in Appendix~\ref{app:Choice of optimal penalty parameter} confirm that the elimination process closely follows the order of parameter estimates in the absence of very strong correlation patterns among the covariates.
\par 
\begin{figure}[h]
    \centering
    \includegraphics[width=0.7\textwidth]{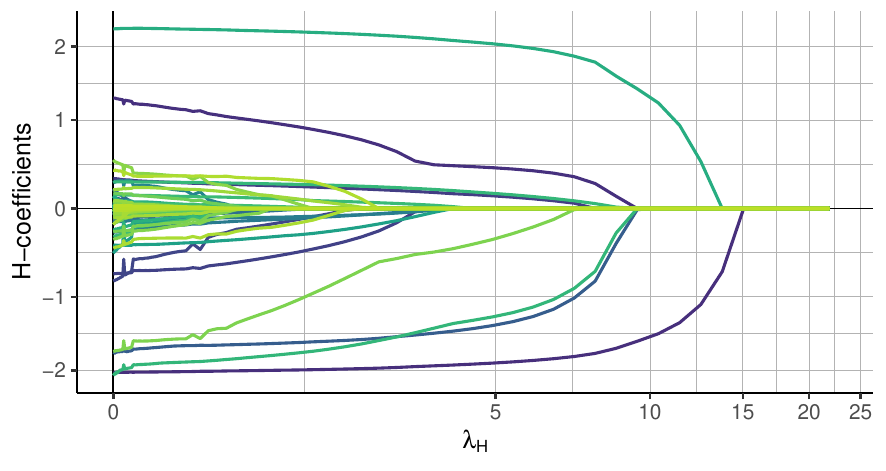}
    \caption{Path of estimated coefficients in $\hat{\Hbf}$ for different values of regularization parameter $\lambda_{H}$}
    \label{fig:fan_H} 
\end{figure}

As discussed in Section~\ref{sec:Data-driven model selection}, the optimal values of the penalty parameters $\left(\lambda_{L}, \lambda_{H}, \lambda_{\beta}\right)$ are chosen by cross-validation. Figure~\ref{fig:lambda_H_cv} shows the mean prediction error on the test sample over the cross-validation folds for different values of the regularization parameter $\lambda_H$. The convex optimization along the penalization parameter clearly shows itself. \par
\begin{figure}[h]
    \centering
    \includegraphics[width=0.7\textwidth]{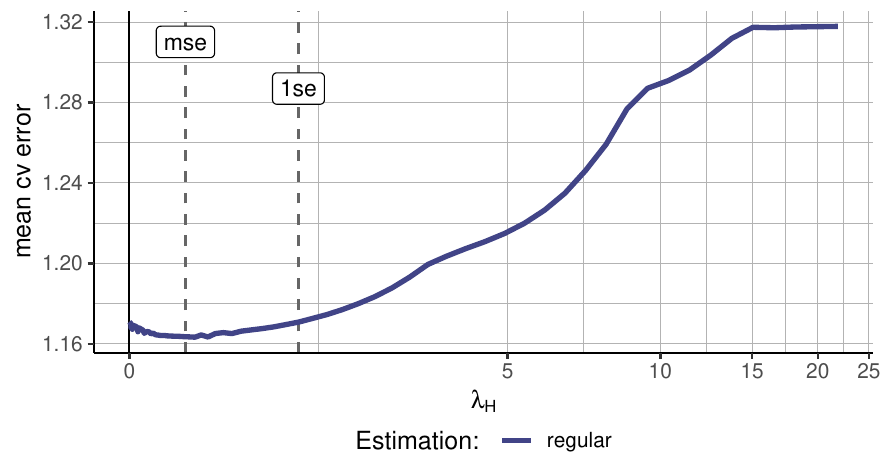}
    \caption{Mean error of out-of-sample prediction over cross-validation folds of $\hat{\Hbf}$ for different values of regularization parameter $\lambda_H$ \newline{\small Note: 'mse' and '1se' show determined $\lambda_H$ based on \textit{mse} and \textit{1se} optimality criteria.}}
    \label{fig:lambda_H_cv}
\end{figure}
As in each fold, the objective function is only evaluated on a subsample of the available data, the cross-validation optimization is not necessarily optimal for the out-of-sample observations. A known issue is the estimation error of risk curves using cross-validation samples \citep{hastie2009elements}. Because of this estimation error, model selection using cross-validation tends to be too conservative while in fact, the smallest model is preferred among a set of indistinguishable models~\citep{Chen2021}. To address this, we implement the common \textit{1se} solution, which selects for the largest value of the penalization parameter for which the objective function remains below the minimum value plus one standard error over the cross-validation folds calculated at the optimal $\lambda_H$ position. The results on the number and size of the non-zero coefficients in Appendix~\ref{app:Choice of optimal penalty parameter} strongly support the application of the \textit{1se} solution.\par 

\subsection{Accuracy of treatment effect estimates}
\label{sec:Accuracy of treatment effect estimates}
In this subsection, we evaluate the estimation error over multiple samples to remove the dependency of the measured estimation error on a specific randomly generated sample. The plots in the subsequent subsections use the abbreviation as described in Table~\ref{tab:Model_abbreviations} to denote different estimation methods.\par

\begin{table}[ht]
	\centering
	\caption{Abbreviation for estimation methods}
	\begin{tabularx}{0.9\textwidth}{lX}
        \toprule
		\textit{no\_reg} & Estimation as in \cite{athey2021matrix} without any covariate regularization. \\
        \textit{imp0} & Estimation with covariate regularization with imposed null hypothesis.\\
        \textit{imp0\_rot} & \textit{imp0} estimation using the rule-of-thumb correction.\\
		\textit{imp0\_post} & Two-stage estimation with model/rank-selection from 'imp0' estimation and unregularized post step.\\
		\textit{imp0\_1se} & \textit{imp0} estimation using '1se' optimality criterion in cross-validation.\\
		\textit{not0} & Estimation with covariate regularization without imposing the null hypothesis. (Does not allow for inference on effect estimates)\\
		\bottomrule
	\end{tabularx}
	%\caption*{\footnotesize Note: For better readability, we omit the subscripts of $e_{t,i}$ if a notation is not tied to a particular event but holds for any arbitrary event $e$.}
	\label{tab:Model_abbreviations}
\end{table}
Figure \ref{fig:sim_boxplot_sample_size_tau} shows the boxplot of ATET estimates by different sample sizes for the six versions of the matrix completion estimator as outlined in Table~\ref{tab:Model_abbreviations}. All estimators using covariate regularization perform a great deal better, in particular for small sample sizes. However, the versions imposing the null hypothesis exhibit a substantial downward bias, which is the price paid for enabling the inference procedure. With a decent sample size, the rule-of-thumb correction is able to eliminate the bias. The matrix completion estimator without imposed null, as in \cite{athey2021matrix}, but using covariate regularization is unbiased and performs already excellently for very small sample sizes. \par
The fit of the post-regularization model is consistently very poor. It can be shown that the estimates of $\hat{\Lbf}$ are closely tied to the regularization by the penalty term. When dropping the penalization parameter in the post-regularization step, the coefficient in $\hat{\Lbf}$ strongly diverge from the true values.\par
\begin{figure}[h]
    \centering
    \includegraphics[width=0.7\linewidth]{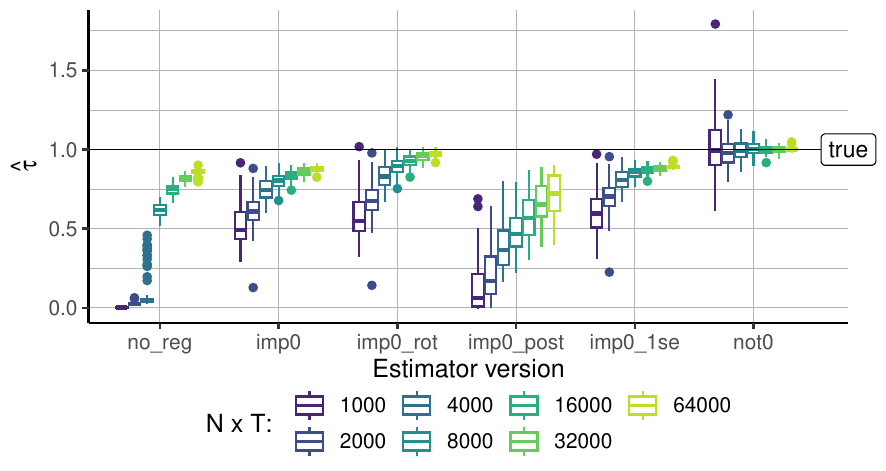}
    \caption{Boxplot of $\hat{\tau}$ by estimation method.\newline{\small Note: $N=100$, $T \in \{10, 20, 40, 80, 160, 320, 640\}$. 700 runs per sample size. True value of $\tau=1$.}}
    \label{fig:sim_boxplot_sample_size_tau} 
\end{figure}
To undertake the estimation accuracy in a more detailed evaluation, Figure~\ref{fig:sim_sample_size_mise} shows the median indexed squared error of the treatment effect where the squared errors in each sample are divided by the squared estimation error of the baseline \textit{imp0} estimator. Hence, a value below 1 denotes that an estimator performs on average better than the basic regularized estimator with imposed null.\footnote{Note that the median indexed squared error is an aggregate of a ratio and has to be interpreted accordingly.} This illustration refines and accentuates the insights from Figure~\ref{fig:sim_boxplot_sample_size_tau} that the enabled inference procedure by imposing the null hypothesis comes at a substantial price in treatment effect estimation accuracy. With increasing sample size, the rule-of-thumb correction achieves similar precision as the matrix completion estimator without the null hypothesis imposed. Thus, if inference on the treatment effect estimator is desired, applying the rule-of-thumb correction can partially compensate for the accuracy loss by the implied null hypothesis for the inference procedure. \par
\begin{figure}[h]
    \centering
    \includegraphics[width=0.7\textwidth]{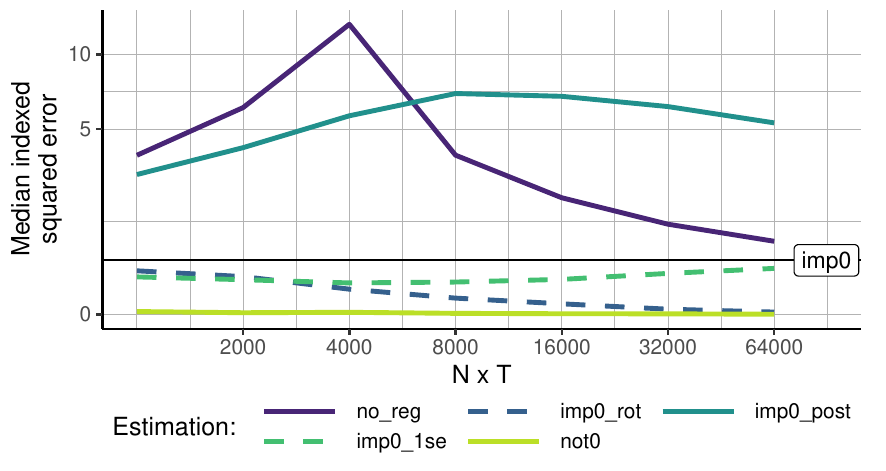}
        \caption{Median indexed squared error of $\hat{\tau}$ by estimator version.\newline{\small Note: \textit{Indexed} measures are divided by the squared estimation error of the \textit{imp0} model. $N=100$, $T \in \{10, 20, 40, 80, 160, 320, 640\}$. 700 runs per sample size. Transformed y-axis.}}
        \label{fig:sim_sample_size_mise}
\end{figure}

The simulation results further evince that the reduction in model size confers benefits to estimation accuracy, particularly in the context of small sample sizes. The lower degrees of freedom of the fitted model result in better predictions for the potential outcomes under no-treatment of the treated observations with the available data. For larger sample sizes, there is enough information in the data such that fitting the full model without covariate regularization picks up less spurious correlations of non-explanatory covariates such that the estimation accuracy levels with the regularized version of the estimator. Nevertheless, the gains in accuracy achieved through by applying the estimate correction or performing an estimation without inference are still massive compared to the unregularized estimation. \par

In Appendix \ref{app:sim_signal_strength}, we demonstrate that the bias of the treatment effect estimates is not negatively affected by a low signal-to-noise ratio, in fact, the bias is slightly decreasing for weaker signals. However, for a low signal-to-noise ratio, the variance of the treatment effect estimate is substantial. The relative performance of the different estimator versions reflects the pattern observed in Figures \ref{fig:sim_boxplot_sample_size_tau} and \ref{fig:sim_sample_size_mise}:
The estimator without imposing the null hypothesis is superior for all settings and the rule-of-thumb correction \textit{imp0\_rot} is consistently the best-performing estimator among the null hypothesis versions, in particular for low signal-to-noise ratios.

\subsection{Model selection property}
\begin{figure}[h]
    \centering
    \includegraphics[width=0.7\textwidth]{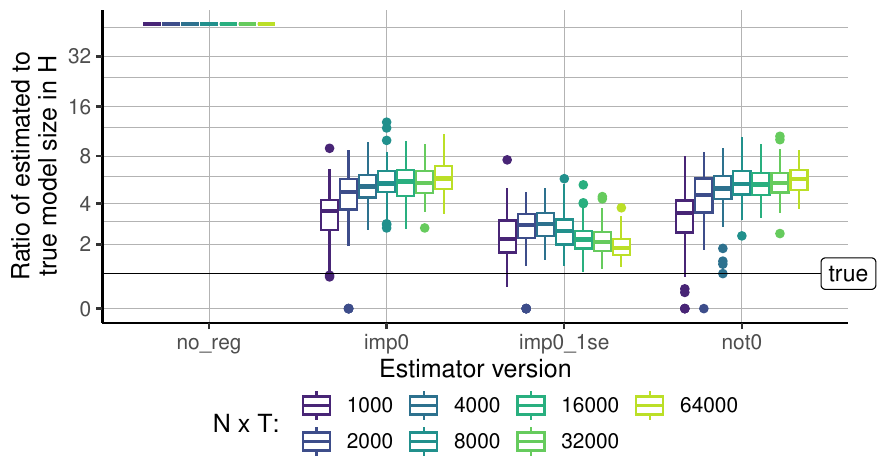}
    \caption{Ratio between the estimated and true model size in $\Hbf$.\newline{\small Note: $N=100$, $T \in \{10, 20, 40, 80, 160, 320, 640\}$. 700 simulations per sample size. The number of non-zero elements is equal for regular, rule-of-thumb correction, and post-regularization estimation. Transformed y-axis.}}
    \label{fig:sim_boxplot_sample_size_H_size_ratio} 
\end{figure}
We assess the dimensions of the ascertained model, defined by the count of non-zero coefficients in the matrix $\Hbf$, in comparison to the true size of the generated sample across varying sample sizes. It is pertinent to note that the post-regularization and rule-of-thumb correction estimates have no impact on the model size. Consequently, we display only the unregularized estimation, the regularization approach employing \textit{mse} and \textit{1se} optimality criterion, and the estimator not imposing the null hypothesis.\par

Figure~\ref{fig:sim_boxplot_sample_size_H_size_ratio} illustrates a substantial reduction in the number of non-zero coefficients within the estimated matrix $\Hbf$ when utilizing a regularized estimator.\footnote{It is worth acknowledging that the actual count of non-zero coefficients in $\Hbf$ may deviate from the DGP parameter $h_{\text{size}}$ under conditions where $p<\bar{p}$ or $q<\bar{q}$.} Notably, the MSE-optimal cross-validation results in an overestimation, maintaining approximately five times the number of covariates within the determined model across all sample sizes. Conversely, the 1se variant consistently selects a model size that closely aligns with the true model size. As the sample size increases, there appears to be a convergence of the estimated model size towards the true value.\par

Additionally, in Appendix \ref{App:H_mse_sample_size}, we provide evidence that the MSE of the estimated coefficients in $\Hbf$ exhibits a strong decrease as the sample size diminishes. The concurrent enhancements in both model size and coefficient accuracy contribute to the augmented precision in treatment effect estimates for larger sample sizes which has been observed in Section \ref{sec:Accuracy of treatment effect estimates}.

In Appendix \ref{app:model_selection_signal_strength}, we illustrate that the determined model size persists is stable as the signal exhibits substantial strength. Nevertheless, as the signal-to-noise ratios approach very small magnitudes, rendering the task of segregating relevant information from noise more difficult, all estimators tend to identify a reduced number of informative parameters within the model. This discernment, in turn, can lead to estimated models that are more sparse than the true model.

\section{Illustration}\label{sec:Illustration}
To showcase the usefulness of the proposed estimator in an empirical context, we apply the developed methodology to scrutinize the impact of governmental regulations aimed at mitigating the Covid-19 pandemic. The period encompassing the peak of the global pandemic witnessed an intense public discourse regarding the appropriateness and effectiveness of public health measures designed to curb the spread of SARS-CoV-2. This discourse extended into the academic domain, putting forth extensive evaluations of, e.g.,  governmental responses to the pandemic \citep{Christensen2023Nordic, basseal2023key} or public compliance with health regulations \citep{scandurra2023people}.\par

In this illustrative analysis, we delve into the effect of international travel restrictions, a measure identified as highly effective by \cite{basseal2023key}. Our focus is on providing a pure analytical assessment of travel restrictions as a public health intervention. The geopolitical implications of travel bans and the suspension of visa exemptions are thoughtfully explored in \cite{seyfi2023covid}.\par

We analyze the impact of Covid-19 testing obligations upon entry from foreign countries with elevated incidence rates on the frequency of infections necessitating treatment in intensive care units (ICU). The testing obligation at entry from high-risk regions was partially implemented in Germany from the summer of 2020 to the spring of 2022.\footnote{The authority for public health regulations in Germany has been dispersed to various administrative levels with a general trend to transition from individual districts to state-harmonized policies over time. Some regions temporarily intensified regulations by imposing testing obligations on all incoming travellers, which is not being distinguished in this analysis.} To circumvent distortions in the outcome measure arising from limited testing capacities during the early stages of the pandemic, we confine the sample period to July 2020 to June 2022, utilizing weekly frequency. Additionally, to ensure robustness in our outcome measure, we exclude districts with fewer than 10 intensive care beds. It is noteworthy that, in the remaining districts, no district operated at maximum ICU occupancy for a significant duration, ensuring consistent reporting and treatment within the same district. With these restrictions, our analysis spans a panel of 342 units over 105 time periods.\par

The adoption of the ratio between patients treated in ICU due to Covid-19 infections and the number of reported infections as the outcome is appealing for two reasons. Firstly, travel restrictions primarily aim to avert the overwhelming strain on medical facilities by delaying the introduction of new and potentially threatening virus mutations. This delay allows more time for the refinement of vaccines to address emerging variants. Secondly, utilizing the frequency of severe infections as the outcome mitigates common concerns related to endogeneity and reverse causality when using plain incidence rates or vaccination rates as outcomes.\footnote{The median duration of ICU treatment for Covid-19 patients is between 10 to 14 days and exhibits a right-skewed distribution \citep{shryane2020length, kaccmaz2023covid}. Consequently, during periods of declining infection rates, the ratio of patients admitted to the ICU to the reported incidence may exceed 1. A direct mapping of ICU patients to an infection is not possible. Consequently, the metric employed in this study should be interpreted as a proxy for the severity of infections. For a more comprehensive examination of the epidemiological landscape extending this illustrative example, the adoption of more sophisticated metrics becomes requisite.}\par

Given the extensive research on the medical, economic, and social impact of the pandemic, comprehensive macroeconomic data is available for the selected time frame. The German Corona-Datenplattform\footnote{\url{https://www.healthcare-datenplattform.de/}} serves as a primary data source, offering a rich collection of macroeconomic, infection, and policy data. The dataset encompasses 91 unit-time-specific covariates (encompassing various public health regulations, vaccination prevalence, and short-time work data), 55 unit-specific covariates (encompassing socio-demographics, medical care, infrastructure, and economic composition variables), and 11 time-specific covariates (covering general economic indicators, mobility, tourism, and the prevalence of threatening mutations among all reported infections). The complete datasets, along with a detailed description of all variables, can be accessed on the Harvard Dataverse \citep{DVN/JGGBQG_2024} related to this project. \par

The provided context and dataset serve to exemplify the estimation and model selection characteristics of the proposed estimator within an empirical application. The application of the panel data model \eqref{eq:Full model}requires the estimation of 1143 parameters.\footnote{An extension of the model to \eqref{eq:Full model linear}, which incorporates linear terms in $\Xbf$ and $\Zbf$, expands the parameter count to 10,678. While technically feasible, estimating this augmented model is deemed impractical due to the exceedingly small signal-to-noise ratio within the specified application setting.} Table \ref{tab:Application_results} summarizes the estimation outcomes and resultant model sizes across the various versions of the estimator, denoted by previously introduced abbreviations detailed in Table \ref{tab:Model_abbreviations}. Notably, the estimated treatment effects are very close to 0. The p-values, approaching unity for the regularized estimators with imposed null hypothesis, suggest an increased precision in estimating potential outcomes for treated observations. This poses a challenge to the regularity assumption of the stochastic shock, as outlined in Assumption \ref{ass:regularity of stochastic shock}. Given the non-rejection of the hypothesis asserting the absence of a treatment effect, there is no compelling argument to accord greater credibility to the results obtained through the \textit{not0} estimation procedure. \par

\begin{table}[ht]
	\centering
	\caption{Estimation results and model sizes by estimator version}
	\begin{tabularx}{0.9\textwidth}{lYYYYY}
         & \textit{no\_reg} & \textit{imp0} & \textit{imp0\_1se} & \textit{not0} & \textit{not0\_1se}\\
        \toprule
        ATET & 0.00007 & 0.00014 & 0.00550 & -0.00052 & 0.00674\\
        $\text{ATET}_\text{rot}$ & 0.00007 & 0.00015 & 0.00586 & - & - \\
        p-value & 0.496 & 0.98 & 1 & - & . \\
        \midrule
        $\hat{\Hbf}$ size & 605 & 107 & 1 & 2 & 1 \\
        $\hat{\bm{\beta}}$ size & 91 & 0 & 0 & 0 & 0 \\
        $\hat{\Lbf}$ rank & 25 & 76 & 1 & 71 & 1 \\
		\bottomrule
	\end{tabularx}
    \caption*{\footnotesize Note: The model size is determined by the number of non-zero elements in $\hat{\Hbf}$ and $\hat{\bm{\beta}}$, as well as the number of singular values of $\hat{\Lbf}$.}
	\label{tab:Application_results}
\end{table}

The lower segment of Table \ref{tab:Application_results} elucidates the resultant model sizes derived from the estimation process. Notably, the \textit{no\_reg} estimation, which solely regularizes the rank of the unobserved factor matrix $\Lbf$ while leaving covariates unpenalized, manifests the full model in terms of $\Hbf$ and $\bm{\beta}$ parameters.\footnote{Recall that the fixed effects remain unregularized throughout.} Estimators integrating covariate regularization substantially reduce the model size, yielding an exceedingly sparse model with merely one covariate when employing the \textit{1se}-criterion. This observation suggests the presence of very weak signal influencing the severity of Covid-19 infection trajectories within the variables of the dataset. The single parameter persisting the model selection process links the proportion of students and the proportion of infections attributed to the Alpha variant (B.1.1.7) of the Covid-19 virus. These findings hint at a notable association between this mutation and an increased incidence of patients necessitating intensive care unit (ICU) treatment, with young individuals in educational settings possibly serving as significant channel for viral transmission. However, it is essential to underscore that this interpretation hinges solely on the interpretation of model parameters and does not denote an identified causal effect.

\section{Conclusion}\label{sec:Conclusion}
We present novel findings concerning $l_1$ covariate regularization within matrix completion methods. A comprehensive simulation study illustrates that this form of regularization significantly diminishes the model size. The determination of penalization parameters through cross-validation, employing the \textit{1se} optimality criterion, yields a model size that converges to the true value as the panel size increases. \par

Moreover, we establish the applicability of the permutation-based inference procedure proposed by \cite{chernozhukov2021exact} to the extended estimator incorporating covariate regularization. Additionally, we demonstrate its validity under any treatment assignment mechanism. While enforcing the null hypothesis of a treatment effect absence, a prerequisite for applying the inference procedure, introduces a downward bias in effect estimates, simulation results indicate that this bias can be substantially reduced, and even completely mitigated for larger sample sizes, through a simple rule-of-thumb correction. \par

The proposed estimator, utilizing both \textit{1se}-regularization and the estimate correction when inference is necessary, displays superior prediction accuracy compared to the baseline matrix completion estimator introduced by \cite{athey2021matrix}. The latter has been demonstrated to outperform other common estimators in panel data regression methods. Apart from improved treatment effect estimates, our proposed estimator additionally features a robust model selection property and enables valid finite sample inference. \par

Future research avenues may include enhancing the two-stage procedure by integrating a more suitable post-model-selection estimation. The application of matrix completion methods in the second step, as evidenced by our simulation study, results in inferior prediction accuracy. Another potential direction for future work involves incorporating an alternative inference procedure that does not impose bias on treatment effect estimates. For instance, a Bayesian inference approach simulating posterior draws via a Markov Chain Monte Carlo sampler, as proposed by \cite{tanaka2021bayesian}, could be explored. \par

Additionally, our proposed estimator can be linked to the broader literature that facilitates the application of matrix completion methods to patterns of missing data that are not completely random, as explored by \cite{bhattacharya2022matrix}, \cite{agarwal2021causal}, and \cite{bai2021matrix}.

\bibliography{mcms_bib.bib}

\newpage
\appendix

\section{Background on matrix completion estimators in the context of causal panel data models}\label{app:Matrix completion estimator}

The econometric literature on panel data with binary treatment exposure has recently evolved in four main streams: The literature on unconfoundedness~\citep{rosenbaum1983central} imputes data for missing potential outcomes using observed outcomes of comparable units in previous periods (See~\cite{arkhangelsky2022doubly} for recent advances). The literature on synthetic control \cite{Abadie2003Synthetic, abadie2010synthetic, amjad2018robust, ben2021augmented, Kellogg2021matching} imputes missing potential outcomes by creating a hypothetical but representative control unit. The literature on factor (or interactive effects) models estimates unobserved outcomes as the sum of a low-rank factor structure and a linear function of covariates~\citep{bai2003inferential, xu2017generalized, Fan2021Recent}. Finally, the difference-in-difference literature relies on the assumption that all average outcomes follow a parallel path over time in the absence of treatment~\citep{sant2020doubly, de2020two,callaway2021difference, goodman2021difference, arkhangelsky2021synthetic}.\par
% alternative methods for panel data estimation
% In a recent approach \cite{goldin2022forecasting} adapts a deep neural architecture for time series architecture to predict unobserved potential outcomes for treated observations.

The suitability of the respective methods in a particular application is mainly determined by the structure of the data. For instance, the unconfoundedness literature typically applies to settings with a single treatment period, making it suitable for \textit{horizontal} regression. On the other hand, the synthetic control literature is well-suited for settings with one (or few) treated units, making it a \textit{vertical} regression framework. In horizontal regression, the focus is on stable patterns over time across units, while vertical regression relies on stable patterns over units across time.\par

The matrix completion method for causal panel data models builds upon the literature on factor models and demonstrates that both horizontal and vertical regression can be nested into the matrix completion objective function~\citep{athey2021matrix}. Similar to general factor models, the outcome model consists of an unobserved low-rank matrix plus noise, with an optional linear covariate component. However, matrix completion methods \citep{mazumder2010spectral, candes2010matrix, candes2012exact, Fernandez2021matrix} differ from factor model literature in how they determine the rank. While factor models estimate, bound, or assume the rank to be known, matrix completion methods determine the rank through regularization, using a penalty term in the objective function. Regularization is particularly suitable in causal panel data settings as the focus is on imputing missing elements in the matrix of potential outcomes rather than consistently estimating the underlying factors. \par 
%The concept of rank regularization in matrix completion is closely related to the principal component analysis of the adjusted covariance matrix in latent factor models \citep{XIONG2023271}.

Matrix completion methods are typically constrained to cases where each observation has an independent and non-zero probability of being missing \citep{ma2019missing}. However, \cite{athey2021matrix} derived bounds for estimating the average treatment effect of the treated (ATET) in causal panel data models, even when the missing data patterns are not completely random.\par

To address the inference challenges across the various panel data literature streams, \cite{chernozhukov2021exact} introduced a uniformly valid inference procedure applicable to matrix completion methods as well. Their permutation-based inference offers a generic and robust approach for settings with few treated units, with the treatment occurring uniformly for a defined number of periods at the end of the time interval, while accommodating stationary and weakly dependent shocks.  \par

The efficiency of regularization on the nuclear norm of the unobserved factor matrix in matrix completion methods is closely related to the number of missing elements in the matrix. For instance, the convergence bounds in \citet{athey2021matrix} assume a maximum probability for missing elements per time-period, while \cite{chernozhukov2021exact} requires the number of pre-treatment periods to be larger than the number of periods under treatment. In most applications, more observations are observed in a treated state, leading to a focus on estimating the average treatment effect of the treated (ATET) in causal analysis using matrix completion methods. However, in cases where there is a larger set of observations under treatment, the perspective can be switched accordingly.

\section{Proofs}
\subsection{Additional notation}
We introduce additional notation that will be used in the proofs. For $a,b\in \Rbb: a\vee b \coloneqq \max(a,b)$. Without any subscript, $\norm{\:\cdot\:}$ denotes the Euclidian norm for vectors and the spectral norm for matrices. $\floor{a}$ rounds $a$ down to the nearest integer. $\floor{a}_b$ denotes a $\pmod{b}$. The element-wise matrix multiplication (Hadamard product) is written as $A \circ B$. The maximal amount of treated time periods per observation is defined as $\card{\Mmc}_{\max} =\max_{1\leq i\leq N} \sum_{1\leq t\leq T} \;(i,t) \!\in\! \Mmc$. $\mathcal{Z}$ denotes the set $\{\Lbf, \Xbf, \Hbf, \Zbf, \Vbf, \mathbf{\beta}\}$. We rewrite the term $\Xbf\Hbf\Zbf$ as a matrix-vector multiplications $\Xbf\Hbf\Zbf=\left[\XZbar_{it}\Hbar\right]_{it}$, where the elements in $\Hbar$ are the concatenated columns of $\Hbf$ and the entries in $\XZbar_{it}$ follow from the equality. \par
Note that, to shorten the notation of the subsequent proofs, we redefine the underlying model~\eqref{eq:Full model} by including the fixed effects term in the factor matrix $\bar{\Lbf}\coloneqq\Lbf+ \mathbf{\Gamma}^*\Ibf_T^\top + \Ibf_N(\mathbf{\Delta}^*)^\top$. This does not affect the low-rank assumption on $\Lbf$. For ease of readability, we misuse the notation of $\Lbf$ by meaning $\bar{\Lbf}$. Adaptations to the proofs using the true model~\eqref{eq:Full model} are straightforward. 

\subsection{Proof of Theorem~\ref{th:pvalue}}
The proof follows directly from Lemma~H.1 in \cite{chernozhukov2021exact}. It remains to verify the approximate ergodicity condition (E) and the estimation error condition (A) of Lemma~H.1 in \cite{chernozhukov2021exact} hold in the present setting with additional $l_1$ regularization on the covariate space and treated observations across the whole units/time-space.\par
Let $n=\card{\Pi}$ and $\delta_{1n},\delta_{2n},\gamma_{1n}, \gamma_{2n}$ be sequences of numbers converging to 0.
\begin{itemize}
    \item[(E)] With probability $1-\gamma_{1n}$ the randomization distribution $ \hat{F}(x)=\frac{1}{n}\sum\nolimits_{\pi \in \Pi} \Ibf\left\{S\left(\hat{U}_\pi\right)<x\right\}$ is approximately ergodic for $F(x)=P[S(u(<x)$, namely $\sup_{x\in\Rbb}\card{\hat{F}(x)-F(x)}\leq \delta_{1n}$.
    \item[(A)] With probability $1-\gamma_{2n}$, the estimation errors are small:
    \begin{enumerate}
        \item The mean squared error is small with $n^{-1} \sum\nolimits_{\pi \in \Pi} \left(S(\hat{U}_\pi)-S(U_\pi)\right)^2\leq \delta_{2n}^2$;
        \item The pointwise error at $\pi=\text{Identity permutation}$ is small with $\card{S(\hat{U})-S(U)}\leq \delta_{2n}$;
        \item The pdf of $S(U)$ is bounded above by a constant D.
    \end{enumerate}
\end{itemize}
The following Lemmas \ref{lem:ergodicity_iid}-\ref{lem:estimation_block} verify the conditions (E) and (A), each for the iid permutations and the moving block permutations. Notice that condition (A.3) is already satisfied by Assumption~\ref{ass:test_statistic}. 

\begin{lemma}[Ergodicity condition (E) for iid permutations]\label{lem:ergodicity_iid}
    Suppose that Assumption~\ref{ass:regularity of stochastic shock}.1 and ~\ref{ass:test_statistic} hold, and let $\Pi$ be the set of all permutations. If $\card{\Mmc}<\card{\Omc}$, then $$P\left[\sup_{x\in\Rbb} \card{\hat{F}(x)-F(x)}\leq \delta_{1n}\right]\geq 1-\gamma_T,$$ where $\gamma_T = \sqrt{\pi/(2+2\floor{\card{\Omc}/\card{\Mmc}})}/\delta_{1n}$
\end{lemma}
\begin{proof}
Recall that in the case of Assumption \ref{ass:regularity of stochastic shock}.1, $\Pi$ is the set of all bijections $\pi$ on $\{1,\dots,T\}$. Let $b(x):x\in \Rbb^{\card{\Omc} \mapsto \Omc}$ be any enumeration of the elements in $\Omc$. Let $k_\Mmc=\floor{\card{\Omc}/\card{\Mmc}}$ and define the set of indices 
$$b_i = \begin{cases}
    \Mmc & i=0\\
    \{b\left((i-1)*\card{\Mmc}+1\right), \dots, b\left((i-1)*\card{\Mmc}+\card{\Mmc}\right)\} & i=1, \dots, k_\Mmc.
\end{cases}$$
Since $S(U)$ only depends on $b_0$, we can define $$Q(x;U_{b_0})=\Ibf\{S(U_{b_0})\leq x\}-F(x).$$
Therefore $$\tilde{F}(x)-F(x)=\card{\Pi}^{-1}\sum_{\pi\in\Pi} Q(x;U_{\pi(b_0)}).$$
Because $Q(x;u)$ only depends on the residuals in $\Mmc$, the value of $\sum\nolimits_{\pi\in\Pi}Q(U_{\pi(b_i)})$ does not depend on $i$. Hence, we can write
\begin{align*}
    \tilde{F}(x)-F(x)&=\card{\Pi}^{-1}\sum_{\pi\in\Pi} Q(x;U_{\pi(b_0)})\\
    &=k_\Mmc^{-1} \sum_{i=0}^{k_\Mmc} \left(\card{\Pi}^{-1}\sum_{\pi\in\Pi} Q(x;U_{\pi(b_i)})\right)\\
    &=\card{\Pi}^{-1}\sum_{\pi\in\Pi}  \left(k_\Mmc^{-1} \sum_{i=0}^{k_\Mmc}Q(x;U_{\pi(b_i)})\right)
\end{align*}
By Jensen's inequality \begin{equation}
    \left[\sup_{x\in\Rbb}\card{\tilde{F}(x)-F(x)}\right] \leq \card{\Pi}^{-1}\sum_{\pi\in\Pi} E\left[\sup_{x\in\Rbb}\card{k_\Mmc^{-1}\sum_{i=0}^{k_\Mmc} Q(x; U_{\pi(b_i)})}\right] \label{eq:jensen}
\end{equation}
We observe that for any $\pi \in\Pi$
\begin{align}
E\left[\sup_{x\in\Rbb}\card{k_\Mmc^{-1}\sum_{i=0}^{k_\Mmc} Q(x; U_{\pi(b_i)})}\right]&=\int_0^1 P\left[\sup_{x\in\Rbb}\card{k_\Mmc^{-1}\sum_{i=0}^{k_\Mmc} Q(x; U_{\pi(b_i)})}>z\right]dz\nonumber \\
&\overset{\text{(i)}}\leq 2\int_0^1 \exp(-2(k_\Mmc +1) z^2)dz \nonumber \\
&< 2\int_0^\infty  \exp(-2(k_\Mmc +1) z^2)dz \nonumber \\
&\overset{\text{(ii)}}< \sqrt{\pi/(2+2k_\Mmc)}\label{eq:dvoretsky},
\end{align}
where (i) follows by the Dvoretsky-Kiefer-Wolfwitz inequality (e.g. Theorem 11.6 in \cite{kosorok2008introduction}) and (ii) is a property of the normal distribution.
Combining \eqref{eq:jensen} and \eqref{eq:dvoretsky} yields $$ \left[\sup_{x\in\Rbb}\card{\tilde{F}(x)-F(x)}\right] \leq \sqrt{\pi/(2+2k_\Mmc)}. $$
The desired result in Lemma~\ref{lem:ergodicity_iid} follows by the Markov's inequality.
\end{proof}

\begin{lemma}[Ergodicity condition (E) for moving block permutations]\label{lem:ergodicity_block}
Suppose that Assumption~\ref{ass:regularity of stochastic shock}.2 and~\ref{ass:test_statistic} hold, and let $\Pi$ be the set of all moving-block permutations. Assume $\card{\Mmc}_{\max}<T$. Suppose that $\forall i: \{u_{i,t}\}_{t=1}^T$ is stationary and strong mixing with $\sum\nolimits_{k=1}^\infty \alpha_\text{mixing}(k)$ is bounded by a constant $M$. Then there exists a constant $c(M)$ such that for any $\delta_{1n}>0$ $$P\left[\sup_{x\in\Rbb} \card{\hat{F}(x)-F(x)}\leq \delta_{1n}\right]\geq 1-\gamma_T,$$ where $\gamma_T = 
\left(c(M)\frac{\sqrt{\card{\Mmc}_{\max}}\log T}{\sqrt{T}}\right)/\delta_{1n}$
\end{lemma}
\begin{proof}
For $\pi \in \Pi$, we define $$s_t=\sum_{(i,t')\in\Mmc}\card{U_{i,\floor{t+t'}_T}}$$ and $$\check{F}(x)=\card{\Pi}^{-1}\sum\nolimits_{1\leq t\leq (T-1)} \Ibf\{s_t \leq x \}.$$
It is straightforward to verify that $\{S(U_{\pi}):\pi\in \Pi \}=\{s_t: 1\leq t\leq T\}$. By Assumption~\ref{ass:regularity of stochastic shock}.2, $\{s_t\}$ is stationary for $t\in \{1,\dots,T-1\}$ but not for $t=T$ and the bounded pdf of $s(U)$ implies the continuity of $F(\cdot)$. Let $\tilde{\alpha}_{\text{mixing}}$ be the strong-mixing coefficients for $\{s_t\}_{1\leq t\leq T}$. Then by Proposition 7.1 of \cite{rio2017asymptotic} it follows that $$E\left[ \sup_{x\in\Rbb} \card{\check{F}(x)-F(x)}^2\right]\leq \frac{1}{\card{\Pi}}\left(1+4\sum_{k=0}^{\card{\Pi}-1}\tilde{\alpha}_{\text{mixing}}(k)\right)\left(3+\frac{\log \card{\Pi}}{2\log 2}\right)^2.$$ Notice that $\tilde{\alpha}_\text{mixing}(k)\leq \card{\Mmc}_{\max}\alpha_{\text{mixing}}(k)$ and hence $\sum_{k=0}^{\card{\Pi}-1}\tilde{\alpha}_{\text{mixing}}(k)\leq \card{\Mmc}\cdot M$. Since by Definition~\ref{def:permutations} the number of possible permutation is $\card{\Pi}=T$, it follows that
$$E\left[ \sup_{x\in\Rbb} \card{\check{F}(x)-F(x)}^2\right]\leq \frac{1}{T}\left(1+4\card{\Mmc}_{\max}\cdot M\right)\left(3+\frac{\log T}{2\log 2}\right)^2.$$
By Liapunov's inequality\begin{align}
   E\left[ \sup_{x\in\Rbb} \card{\check{F}(x)-F(x)}\right]&\leq \sqrt{E\left[ \sup_{x\in\Rbb} \card{\check{F}(x)-F(x)}^2\right] }\nonumber\\
   &\leq \sqrt{\frac{1+4\card{\Mmc}_{\max}\cdot M}{T}}\left(3+\frac{\log T}{2\log 2}\right)\label{eq:EMB_1}
\end{align}
Since $T\tilde{F}(x)-(T-1)\check{F}(x)=\Ibf\{s_0\leq x\}$, it follows that \begin{align}
    \sup_{x\in\Rbb} \card{\tilde{F}(x)-\check{F}(x)}&=\sup_{x\in\Rbb} \card{\left(\frac{T-1}{T}\check{F}(x)+\frac{1}{T}\Ibf\{s_0\leq x\}\right)-\check{F}(x)} \nonumber \\
    &=\sup_{x\in\Rbb} \card{\frac{1}{T}\left(\check{F}(x)+\Ibf\{s_0\leq x\}\right)-\check{F}(x)} \nonumber \\
    &\leq \frac{2}{T},\label{eq:EMB_2} 
\end{align}
as both terms, $\check{F}(x)$ and $\Ibf\{s_0\leq x\}$, are bounded by $1$. Combining \eqref{eq:EMB_1} and \eqref{eq:EMB_2}, we obtain that $$ E\left[ \sup_{x\in\Rbb} \card{\tilde{F}(x)-F(x)}\right]\leq \sqrt{\frac{1+4\card{\Mmc}_{\max}\cdot M}{T}}\left(3+\frac{\log T}{2\log 2}\right) + \frac{2}{T}.
$$ 
The desired result follows by Markov's inequality.
\end{proof}

\begin{lemma}[Estimation error condition (A) for iid permutations]\label{lem:estimation_iid}
    Let $\Mmc$ be fixed. Assume that there exists a constant $Q>0$ such that $\norm{S(U)-S(V)}\leq Q\norm{\Wbf\circ(U-V)}_2$ for any $U,V\in \Rbb^{N\times T}$. Suppose further that Assumption \ref{ass:consistency} holds, then conditions (A.1) and (A.2) are satisfied.
\end{lemma}
\begin{proof}
    For $(i,t),(j,s)\in \{1,\dots,N\}\times\{1,\dots,T\}$, we define $A_{(i,t),(j,s)}=\{\pi\in\Pi : \pi(i,t)=(j,s)\}$. Recall that under Assumption~\ref{ass:regularity of stochastic shock}.1 $\Pi$ is the set of all bijections on  $ \{1,\dots,N\}\times\{1,\dots,T\}$. Thus, $\card{A_{(i,t),(j,s)}}=(NT-1)!$. It follows that for any $(i,t)\in \{1,\dots,N\}\times\{1,\dots,T\}$
    \begin{align}
        \sum_{\pi \in \Pi} \left(\hat{U}_{\pi(i,t)}-U_{\pi(i,t)}\right)^2& = \sum_{(j,s)}\sum_{\pi \in A_{(i,t),(j,s)} } \left(\hat{U}_{\pi(i,t)}-U_{\pi(i,t)}\right)^2\nonumber\\
        &=\sum_{(j,s)}\sum_{\pi \in A_{(i,t),(j,s)} } \left(\hat{U}_{(j,s)}-U_{(j,s)}\right)^2\nonumber\\
        &=\card{ A_{(i,t),(j,s)}}\sum_{(j,s)}\left(\hat{U}_{(j,s)}-U_{(j,s)}\right)^2\nonumber\\
        &=(NT-1)!\norm{\hat{U}-U}_2^2\label{eq:norm_U}
    \end{align}
    By assumption, we have
    \begin{align*}
    \frac{1}{\card{\Pi}}\sum_{\pi\in\Pi}\left(S(\hat{U}_\pi)-S(U_\pi)\right)^2 &\leq \frac{Q}{\card{\Pi}}\sum_{\pi\in\Pi}\norm{\Wbf \circ (\hat{U}_\pi - U_\pi)}_2^2 \\
    &\leq \frac{Q}{\card{\Pi}}\sum_{\pi\in\Pi}\sum_{(i,t)\in\Mmc}\left(\hat{U}_{\pi(i,t)}-U_{\pi(i,t)}\right)^2 \\
    &\leq \frac{Q}{\card{\Pi}}\sum_{(i,t)\in\Mmc}\sum_{\pi\in\Pi}\left(\hat{U}_{\pi(i,t)}-U_{\pi(i,t)}\right)^2 
    \end{align*}
    Using $\card{\Pi}=(NT)!$ and applying \eqref{eq:norm_U} it follows that \begin{align*}
        \frac{1}{\card{\Pi}}\sum_{\pi\in\Pi}\left(S(\hat{U}_\pi)-S(U_\pi)\right)^2 &\leq \frac{Q}{(NT)!}\sum_{(i,t)\in\Mmc}(NT-1)!\norm{\hat{U}-U}_2^2 \\
        &\leq \frac{Q}{(NT)!} \card{\Mmc}(NT-1)!\norm{\hat{U}-U}_2^2 \\
        &\leq \frac{Q \card{\Mmc}}{NT}\norm{\hat{U}-U}_2^2
    \end{align*}
    Since $\card{\Mmc}$ is fixed, condition (A.1) holds by Assumption~\ref{ass:consistency}.1.\par
    Recall that $S(U)$ only depends on the elements of $U$ in the set $\Mmc$. Let $U_\Mmc=\{U_{it} :  (i,t) \in\Mmc$. The Lipschitz property of $S(\cdot)$ implies that 
    \begin{align*}
        \card{S(\hat{U})-S(U)}&= \card{S(\hat{U}_\Mmc)-S(U_\Mmc)}\\
        &\leq k\card{\hat{U}_\Mmc-U_\Mmc} 
    \end{align*}
    Condition (A.2) follows by Assumption~\ref{ass:consistency}.2. The proof is complete.
\end{proof}

\begin{lemma}[Estimation error condition (E) for moving block permutations]\label{lem:estimation_block}
     Let $\Mmc$ be fixed and suppose $\card{\Mmc}_{\max}<T$. Assume that there exists a constant $Q>0$ such that $\norm{S(U)-S(V)}\leq Q\norm{\Wbf\circ(U-V)}_2$ for any $U,V\in \Rbb^{N\times T}$. Suppose further that Assumption \ref{ass:consistency} holds, then conditions (A.1) and (A.2) are satisfied.
\end{lemma}
\begin{proof}
    Notice that for moving block permutations \begin{equation}\sum_{(i,t)\in\Mmc}\sum_{\pi\in\Pi} \left(\hat{U}_{i,\pi(t)}-U_{i,\pi(t)}\right)^2\leq\card{\Mmc}_\text{max}\norm{\hat{U}-U}_2^2,\label{eq:norm_U_perm}\end{equation} since the permutations shift the residuals once over the full time interval. \par
    By assumption, we have 
    \begin{align}
         \frac{1}{\card{\Pi}}\sum_{\pi\in\Pi}\left(S(\hat{U}_\pi)-S(U_\pi)\right)^2 &\leq \frac{Q}{\card{\Pi}}\sum_{\pi\in\Pi}\norm{\Wbf \circ (\hat{U}_\pi - U_\pi)}_2^2 \nonumber \\
         &\leq \frac{Q}{\card{\Pi}}\sum_{\pi\in\Pi}\sum_{(i,t)\in\Mmc} \left(\hat{U}_{i,\pi(t)}-U_{i,\pi(t)}\right)^2 \nonumber \\
         &\leq \frac{Q}{\card{\Pi}}\sum_{(i,t)\in\Mmc}\sum_{\pi\in\Pi} \left(\hat{U}_{i,\pi(t)}-U_{i,\pi(t)}\right)^2 \\
    \end{align}
    Using $\card{\Pi}=T$ and applying \eqref{eq:norm_U_perm} it follows that \begin{equation*}
        \frac{1}{\card{\Pi}}\sum_{\pi\in\Pi}\left(S(\hat{U}_\pi)-S(U_\pi)\right)^2 \leq \frac{Q}{T}\card{\Mmc}_\text{max}\norm{\hat{U}-U}_2^2
    \end{equation*}
    Since $\card{\Mmc}$ is fixed and $\card{\Mmc}_\text{max}<T$, condition (A.1) holds by Assumption~\ref{ass:consistency}.1.\par
    Recall that $S(U)$ only depends on the elements of $U$ in the set $\Mmc$. Let $U_\Mmc=\{U_{it} :  (i,t) \in\Mmc$. The Lipschitz property of $S(\cdot)$ implies that 
    \begin{align*}
        \card{S(\hat{U})-S(U)}&= \card{S(\hat{U}_\Mmc)-S(U_\Mmc)}\\
        &\leq k\card{\hat{U}_\Mmc-U_\Mmc} 
    \end{align*}
    Condition (A.2) follows by Assumption~\ref{ass:consistency}.2. The proof is complete.
\end{proof}

\subsection{Proof of Lemma~\ref{lem:consistency}} Assume that
\begin{itemize}[itemsep=-4pt]
    \item[(L\ref{lem:consistency}.1)] Assumption~\ref{ass:identifying condictions} holds,
    \item[(L\ref{lem:consistency}.2)] $\norm{\Lbf}_*\leq \Lambda_L, \norm{\Hbf}_{1,e}\leq \Lambda_H, \norm{\mathbf{\beta}}_{1,e}\leq \Lambda_\beta$,
    \item[(L\ref{lem:consistency}.3)] $\exists \kappa_1 \: : \: \max_{1\leq j \leq N } T^{-1} \sum\nolimits_{t=1}^T E\left[\card{u_{it}}^{2\kappa_1} \vert \mathcal{Z}\right] = O_P(1)$,
    \item[(L\ref{lem:consistency}.4)] $\norm{N^{-1} \sum\nolimits_{i=1}^N E\left[U_iU_i^\top\vert \mathcal{Z}\right]}=O_P(1)$,
    \item[] there exists a sequence $\nu_{N,T}\geq 0$ such that
    \item[(L\ref{lem:consistency}.5)] $\nu_{N,T} (NT)^{-1}\left(\Lambda_L \sqrt{N \vee (N^{1/\kappa_1}T\log(N)}+PQ\Lambda_Hc_H + J\Lambda_\beta c_\beta\right)=o(1)$
    \item[(L\ref{lem:consistency}.6)] $(NT)^{-1} \sum\nolimits_{i=1}^N \sum\nolimits_{t=1}^T (\Yhoe_{it}-\Yoe_{it})^2\leq \nu_{N,T} (NT)^{-1} \sum\nolimits_{i=1}^N \sum\nolimits_{t=1}^T (\Yhoe_{it}-\Yoe_{it})^2$
    \item[(L\ref{lem:consistency}.7)] $\forall_{(j,s)\in \Mmc}\: : \: \card{\Yhoe_{js}-\Yoe_{js}} \leq \nu_{N,T} (NT)^{-1} \sum\nolimits_{i=1}^N \sum\nolimits_{t=1}^T (\Yhoe_{it}-\Yoe_{it})^2$
\end{itemize}
Define $\Xi\coloneqq \Yhoe-\Yoe$ and $Q_{it} \in \Rbb^{N\!\times\! T} \: : \: [\Qbf_{it}]_{js}=\Ibf\{(j,s)=(i,t)\}$, i.e. the matrix of zeros except for the $(i,t)$ element is 1. We thus can rewrite the model as \begin{equation}\label{eq:trace}
    \Ybf_{it}(0)=\text{trace}(\Qbf_{it}^\top\Yoe)+\Ubf_{it}
\end{equation}
Note that by definition the estimator~\eqref{eq:Estimator} satisfies $$ \sum_{i=1}^N \sum_{t=1}^T \left(\Ybf_{it}(0)-\text{trace}(\Qbf_{it}^\top\Yhoe)\right)^2 \leq \sum_{i=1}^N \sum_{t=1}^T \left(\Ybf_{it}(0)-\text{trace}(\Qbf_{it}^\top\Yoe)\right)^2.$$ 
We plug-in \eqref{eq:trace} and reformulate to obtain
\begin{align}
    \sum_{i=1}^N \sum_{t=1}^T \left(\text{trace}(\Qbf_{it}^\top\Yoe)+\Ubf_{it}-\text{trace}(\Qbf_{it}^\top\Yhoe)\right)^2 &\leq \sum_{i=1}^N \sum_{t=1}^T \left(\Ubf_{it}\right)^2 \nonumber \\
    \sum_{i=1}^N \sum_{t=1}^T \left(\Ubf_{it}-\text{trace}(\Qbf_{it}^\top\Xi)\right)^2 &\leq \sum_{i=1}^N \sum_{t=1}^T \left(\Ubf_{it}\right)^2 \nonumber \\
    \sum_{i=1}^N \sum_{t=1}^T \left(\text{trace}(\Qbf_{it}^\top\Xi)\right)^2 &\leq 2\sum_{i=1}^N \sum_{t=1}^T \Ubf_{it}\text{trace}(\Qbf_{it}^\top\Xi)\nonumber \\
    \sum_{i=1}^N \sum_{t=1}^T \left(\text{trace}(\Qbf_{it}^\top\Xi)\right)^2 &\leq 2\,\text{trace}\left(\left[\sum_{i=1}^N \sum_{t=1}^T \Ubf_{it}\Qbf_{it}\right]^\top\Xi\right)\nonumber \\
    \sum_{i=1}^N \sum_{t=1}^T \left(\text{trace}(\Qbf_{it}^\top\Xi)\right)^2 &\leq 2\,\text{trace}\left(\Ubf^\top\Xi\right)
\end{align}
Usign $\Xi=  (\hat{\Lbf}-\Lbf)+\left[\XZbar_{it}(\hat{\Hbar}-\Hbar)\right]_{it} + \left[\Vbf_{it}^\top(\hat{\bm{\beta}}-\bm{\beta})\right]_{it}$ and applying Lemma~\ref{lem:trace}, it follows that 
\begin{align}
     \text{trace}\left(\Ubf^\top\Xi\right) &= \text{trace}\left(\Ubf^\top(\hat{\Lbf}-\Lbf)\right)+\text{trace}\left(\Ubf^\top\left[\XZbar_{it}(\hat{\Hbar}-\Hbar)\right]_{it}\right)+\text{trace}\left(\Ubf^\top\left[\Vbf_{it}^\top(\hat{\bm{\beta}}-\bm{\beta})\right]_{it}\right)  \nonumber\\
     &\overset{\text{(i)}}\leq \norm{u} \cdot \norm{\hat{\Lbf}-\Lbf}_*+ 2PQ\Lambda_H\delta_H + 2J\Lambda_\beta\delta_\beta  \nonumber\\
     &\leq 2\left(\lambda_L \norm{u} + PQ\Lambda_H\delta_H + J\Lambda_\beta\delta_\beta\right)  
\end{align}
where the first term in (i) follows from the trace duality property. To bound $\norm{u}$, we apply Lemma H.10 of \cite{chernozhukov2021exact} and get
\begin{align}
     \sum_{i=1}^N \sum_{t=1}^T \left(\text{trace}(\Qbf_{it}^\top\Xi)\right)^2&\leq 4\left(\lambda_L \norm{u} + PQ\Lambda_H\delta_H + J\Lambda_\beta\delta_\beta\right)   \nonumber\\
     \sum_{i=1}^N \sum_{t=1}^T \left(\text{trace}(\Qbf_{it}^\top\Xi)\right)^2 &=O\left(\Lambda_L \sqrt{N \vee (N^{1/\kappa_1}T\log(N)}+PQ\Lambda_Hc_H + J\Lambda_\beta c_\beta\right)  \nonumber\\
     (NT)^{-1}\sum_{i=1}^N \sum_{t=1}^T \left(\text{trace}(\Qbf_{it}^\top\Xi)\right)^2 &=O\left((NT)^{-1} \vphantom{\sqrt{N \vee (N^{1/\kappa_1}T\log(N)}}\right. \nonumber \\
             &\phantom{=}\; \left. \left(\Lambda_L \sqrt{N \vee (N^{1/\kappa_1}T\log(N)}+PQ\Lambda_Hc_H + J\Lambda_\beta c_\beta\right) \right) \label{eq:lemma_O}
\end{align}
The first result in Lemma~\ref{lem:consistency} follows from applying Assumptions~(L\ref{lem:consistency}.5) and~(L\ref{lem:consistency}.6) to~\eqref{eq:lemma_O}, the second result by applying Assumptions~(L\ref{lem:consistency}.5) and~(L\ref{lem:consistency}.7).
\begin{lemma}\label{lem:trace}
    Assume that there exist two constants $\delta_H, \delta_\beta$ such that $\max_{1\leq h\leq PQ}\sum_{i=1}^N\sum_{t=1}^T u_{it}\XZbar_{it,h}\leq \delta_H$ and $\max_{1\leq b\leq J}\sum_{i=1}^N\sum_{t=1}^T u_{it}V_{it,b}\leq \delta_H$. Then it holds that $\text{trace}\left(\Ubf^\top\left[\XZbar_{it}(\hat{\Hbar}-\Hbar)\right]_{it}\right)\leq 2PQ\delta_H\Lambda_H$ and $\text{trace}\left(\Ubf^\top\left[V_{it}(\hat{\mathbf{\beta}}-\mathbf{\beta})\right]_{it}\right)\leq 2J\delta_\beta\Lambda_\beta$
\end{lemma}
\begin{proof}
Recall that $\sum_{h=1}^{PQ}\Hbar_h\leq\Lambda_H$. Thus
\begin{align*}
    \text{trace}\left(\Ubf^\top\left[\XZbar_{it}(\hat{\Hbar}-\Hbar)\right]_{it}\right) &= \sum_{i=1}^N\sum_{t=1}^T\sum_{h=1}^{PQ}\varepsilon_{it}\XZbar_{it,h}(\hat{\Hbar}_h-\Hbar_h)\\
    &= \sum_{h=1}^{PQ}(\hat{\Hbar}_h-\Hbar_h)\sum_{i=1}^N\sum_{t=1}^T\varepsilon_{it}\XZbar_{it,h}\\
    &\leq \delta_H\sum_{h=1}^{PQ}(\hat{\Hbar}_h-\Hbar_h)\\
    &\leq 2PQ\delta_H\Lambda_H.
\end{align*}
The proof of $\text{trace}\left(\Ubf^\top\left[V_{it}(\hat{\mathbf{\beta}}-\mathbf{\beta})\right]_{it}\right)\leq 2J\delta_\beta\Lambda_\beta$ is equivalent.
\end{proof}

\section{Formulation of data generating process}
\label{app:Formulation of data generating process}
The data-generating process is defined by the following equations:
\begin{singlespace}
\small{
\begin{align*}
    &\Ybf = \tau \mathbf{W} + \Lbf + \Xbf\Hbf\Zbf + \left[\Vbf_{it}^\top \bm{\beta} \right]_{it} + \mathbf{\Gamma}\Ibf_T^\top + \Ibf_N(\mathbf{\Delta})^\top + \Ubf,\; \Ybf \!\in\!\Rbb^{N\!\times\!T}\\
    &\mathbf{W}_{it} \sim \text{Bernoulli}(w)\\
    & \Lbf = \mathbf{U}(\sigma \cdot\Ibf)\mathbf{V}^\top, \; \sigma_j=\begin{cases}
        \sim \exp(\zeta_L), & j\leq \text{rank}_L\\ 0 & j>\text{rank}_L  \end{cases}, \; (\mathbf{U}, \Sigma, \mathbf{V}^T)=\text{svd}(\Lbf'), \; [\Lbf']_{it}\sim \mathcal{N}(0,1)\\  
    &\Xbf = \left[\eta_1\mathbf{X_1}
    , \cdots, \eta_N\mathbf{X_N}\right]^\top\!\in\! \Rbb^{N\!\times\! p},\; \mathbf{X_i}\sim \left(\mathbf{0},\mathbf{\Sigma_X}\right),\; \left[\mathbf{\Sigma_X}\right]_{ab}=\begin{cases}
        1, & \!a\!=\!b\\ \sim \mathcal{U}(0,\sigma_{\max}), &\! a\!\neq \!b
    \end{cases}, \; \eta_i\sim\mathcal{U}(0,1)\\
    &\Hbf_{it} = \eta_{it}\zeta_{it}, \; \eta_{it}\sim \mathcal{N}(0,h_{\text{size}}),\; \zeta_{it}\sim\text{Bernoulli}(h_\text{prob})\\
    &\Zbf = \left[\eta_1\mathbf{Z_1}
    , \cdots,\eta_T\mathbf{Z_T}\right]\!\in\! \Rbb^{p\!\times\! T},\; \mathbf{Z_t}\sim \left(\mathbf{0},\mathbf{\Sigma_Z}\right),\; \left[\mathbf{\Sigma_Z}\right]_{ab}=\begin{cases}
        1, & \!a\!=\!b\\ \sim \mathcal{U}(0,\sigma_{\max}), &\! a\!\neq \!b
    \end{cases}, \; \eta_t\sim\mathcal{U}(0,1)\\
    &\Vbf_{it} \sim (\mathbf{0},\Ibf), \Vbf_{it}\in \Rbb^B\\
    & \bm{\beta}_b = \eta_{b}\zeta_{b}, \; \eta_{b}\sim \mathcal{N}(0,b_{\text{size}}),\;\zeta_{b}\sim\text{Bernoulli}(b_\text{prob})\\
    &\mathbf{\Gamma}_i \sim \mathcal{N}(0,1)\\
    &\mathbf{\Delta}_t \sim \mathcal{N}(0,1) \\
    &\Ubf_{it} \sim \mathcal{N}(0,\sigma_\text{eps})
\end{align*}
}
\end{singlespace}

If not specified otherwise, the following parameters for the data-generating process are used:
$N=100$, $T=80$, $\tau=1$, $\text{rank}_L=5$, $w=0.1$, $\sigma_{\max}=0.8$, $p=50$, $q=20$, $h_\text{size}=1$, $h_\text{prob}=0.025$, $B=1000$, $b_\text{size}=1$, $b_\text{prob}=0.02$, $\sigma_\text{eps}=1$. The number of folds for the cross-validation is set to 5. To ensure comparability across individual simulation runs, the Bernoulli-distributed random variables can be replaced by draws without replacement and respective ratios for small sample sizes.

\section{Additional simulation results for $\Hbf$ (ASR-$\Hbf$)}
\subsection{ASR-$\Hbf$: Choice of optimal penalty parameter}
\label{app:Choice of optimal penalty parameter}
Figure~\ref{fig:lambda} displays the aggregate level on all coefficients for different values of $\lambda_H$: \ref{fig:lambda_H_norm} shows the element-wise $l_1$ norm of the estimated matrix $\hat{\Hbf}$ and \ref{fig:lambda_H_size} the number of non-zero coefficients. Both aggregation measures coincide with the observed pattern in Figure~\ref{fig:fan_H}. The number of coefficients in the constituted model by specific values of the penalty parameter decreases very fast already for low levels of regularization. In accordance with Figure~\ref{fig:fan_H}, the first coefficients to be eliminated are the ones with a small dedicated parameter estimate, letting the norm of the coefficients decline slower than the number of non-zero coefficients. The last coefficients to be regularized out are the ones with the highest absolute coefficient estimate. \par

Figure~\ref{fig:lambda} further illustrates that the \textit{post}-regularization estimation without the penalty terms in the estimation equation~\eqref{eq:Estimator} and including only the covariates from the selected model by the \textit{regular} matrix completion estimation yields larger coefficient estimates in absolute terms as they are not shrunken by regularization. The values for the post-estimation are not shown in Figure~\ref{fig:lambda_H_size} as its model coincides with the regular estimation by construction. \par

\begin{figure}[h]
    \centering
    \begin{subfigure}[b]{0.9\textwidth} \centering
        \includegraphics[width=0.7\textwidth]{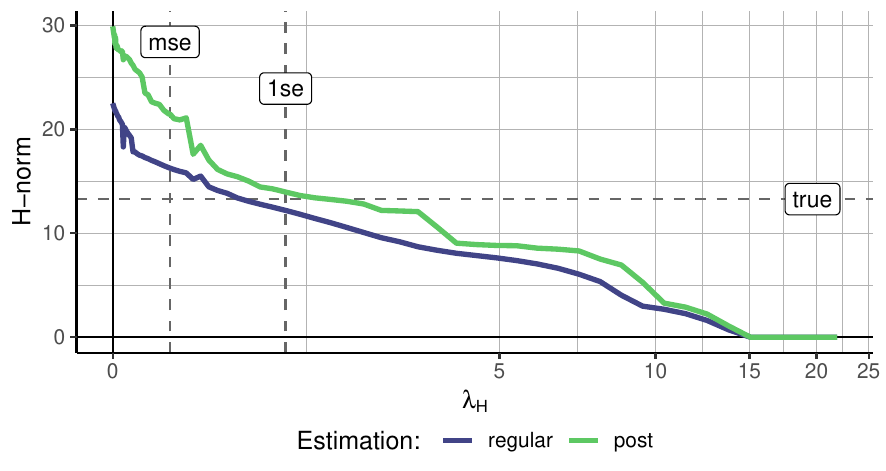}
        \caption{Element-wise $l_1$ norm of $\hat{\Hbf}$}
        \label{fig:lambda_H_norm} 
    \end{subfigure}\vspace{12pt}

    \begin{subfigure}[b]{0.9\textwidth} \centering
        \includegraphics[width=0.7\textwidth]{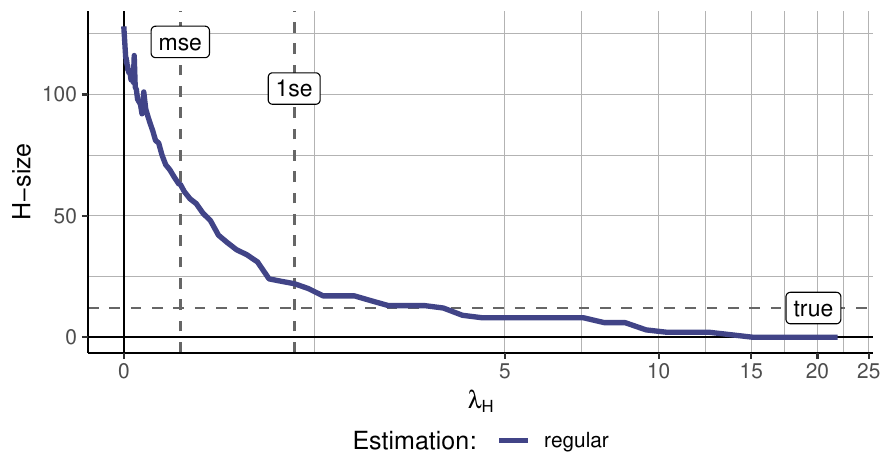}
        \caption{Number of non-zero elements in $\hat{\Hbf}$}
        \label{fig:lambda_H_size}
    \end{subfigure}

    \caption{Element-wise $l_1$ norm (a) and number of non-zero elements (b) of $\hat{\Hbf}$ estimate for different values of regularization parameter $\lambda_H$ \newline{\small Note: 'mse' and '1se' show determined $\lambda_H$ based on \textit{mse} and \textit{1se} optimality criteria, 'true' denotes the true value in the data. The number of non-zero elements is equal for regular and post-regularization estimation}}
    \label{fig:lambda}
\end{figure}

\subsection{ASR-$\Hbf$: Accuracy of treatment effect estimates by signal-to-noise ratio}\label{app:sim_signal_strength}

\begin{figure}[h]
    \centering
    \includegraphics[width=0.7\linewidth]{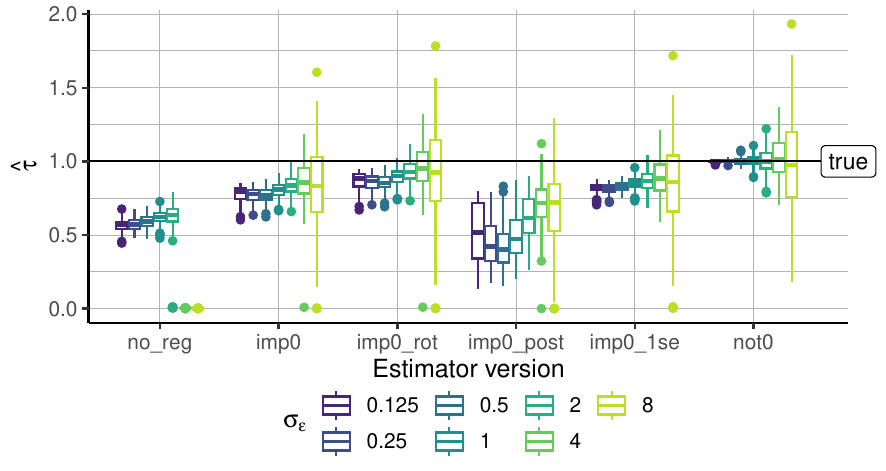}
    \caption{Boxplot of $\hat{\tau}$ by estimation method.\newline{\small Note: Signal-to-noise ratio is given by $1/\sigma_\epsilon$. $N=100$, $T=80$. 1400 runs per signal-to-noise ratio.}}
    \label{fig:sim_boxplot_signal_strength_tau} 
\end{figure}

Figure~\ref{fig:sim_boxplot_signal_strength_tau} shows the estimated treatment effects on the treated for different signal-to-noise ratios. With the true parameters for all explanatory covariates and the fixed treatment effect $\tau$ set to 1, the signal-to-noise ratio is given by the inverse of the error-term standard deviation $\sigma_\epsilon$. All versions of the matrix completion estimator exhibit increasing variance with decreasing signal strength. This is as expected as with larger random shocks, it is more difficult for the estimator to filter out the random components and identify the true signals. The bias pattern of the estimators with the imposed null hypothesis is equivalent to the insights from Section~\ref{sec:Accuracy of treatment effect estimates}. The regularization without imposing the null hypothesis consistently outperforms all other estimators. Among the version with imposed null, the rule-of-thumb correct remains the best-performing estimator, in particular for small signal strengths.\par
\begin{figure}[h]
    \centering
    \includegraphics[width=0.7\textwidth]{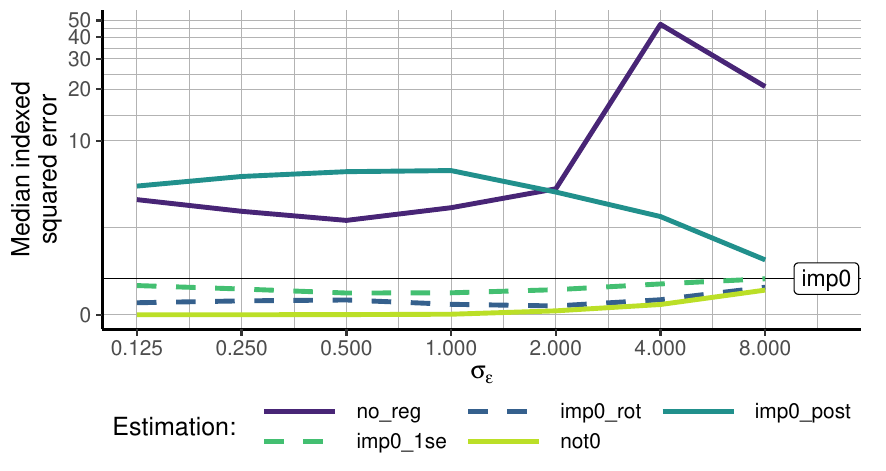}
        \caption{Median indexed squared error of $\hat{\tau}$ by estimation method.\newline{\small Note: \textit{Indexed} measures are divided by the estimate of the \textit{imp0} model. Signal-to-noise ratio is given by $1/\sigma_\epsilon$. $N=100$, $T=80$. 1400 runs per signal-to-noise ratio. Transformed y-axis.}}
        \label{fig:sim_signal_strength_mise}
\end{figure}
In Figure~\ref{fig:sim_signal_strength_mise}, we illustrate the median indexed squared error by the non-regularized estimates, equivalent as described in Section~\ref{sec:Accuracy of treatment effect estimates}. For weak signals, the unregularized estimator heavily struggles to extract the true parameters for its high-dimensional model. All regularized estimator versions perform clearly better as they restrict the focus to a few covariate parameters only. It is worth noting that the differences in accuracy for the estimators with regularization vanish with a decreasing signal-to-noise ratio. 

\subsection{ASR-$\Hbf$: MSE of coefficient estimates in $\Hbf$}\label{App:H_mse_sample_size}

\begin{figure}[h]
    \centering
    \includegraphics[width=0.7\textwidth]{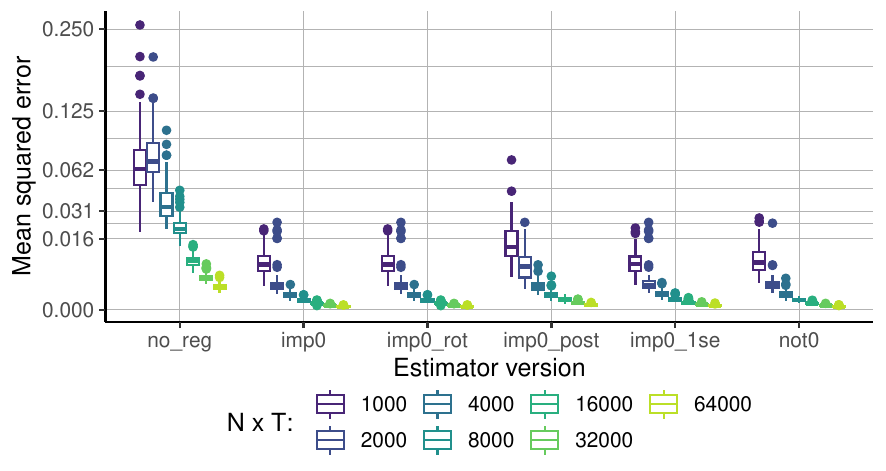}
    \caption{Mean squared error of $\Hbf$-coefficients.\newline{\small Note: $N=100$, $T \in \{10, 20, 40, 80, 160, 320, 640\}$. 700 simulations per sample size. Transformed y-axis.}}
    \label{fig:sim_boxplot_sample_size_H_mse}
\end{figure}
In Figure \ref{fig:sim_boxplot_sample_size_H_mse}, a boxplot representation illustrates the mean squared errors associated with the estimated coefficients across multiple simulation runs. Notably, both the magnitude and variance of coefficient estimation errors exhibit a pronounced decline as the sample size increases. It is observed that post-regularization estimates exhibit the highest MSE among all versions of the estimator considered.
Figure \ref{fig:sim_boxplot_sample_size_H_mse} shows the boxplot of the mean squared errors of the estimated coefficients for each simulation run. The magnitude and the variance of the coefficient estimation errors are rapidly decreasing in the sample size. The post-regularization estimates exhibit the largest mse among all estimator versions. In light of the fact that the second stage produces unbiased estimates through the exclusion of $l_1$-regularization, it follows that the post-regularization coefficient estimates are susceptible to increased variability.\par

\subsection{ASR-$\Hbf$: Model selection property for different signal-to-noise ratios}\label{app:model_selection_signal_strength}
\begin{figure}[h]
    \centering
    \includegraphics[width=0.7\textwidth]{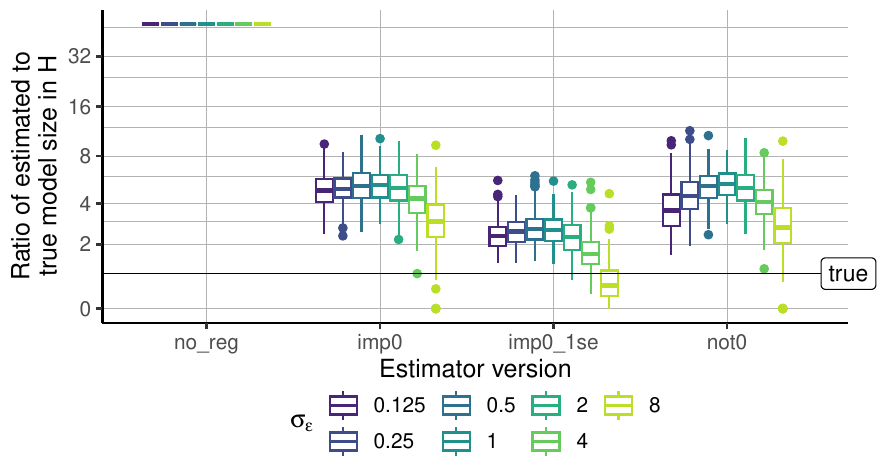}
    \caption{Ratio between the estimated and true model size in $\Hbf$.\newline{\small Note: Signal-to-noise ratio is given by $1/\sigma_\epsilon$. 1400 simulations per signal strength. $N=100$, $T=80$. The number of non-zero elements is equal for regular, rule-of-thumb correction, and post-regularization estimation. Transformed y-axis.}}
    \label{fig:sim_boxplot_signal_strength_H_size_ratio} 
\end{figure}
Figure \ref{fig:sim_boxplot_signal_strength_H_size_ratio} illustrates that the determined model size remains stable when the signal possesses considerable strength. However, for signal-to-noise ratios approaching very small magnitudes, all estimators tend to discern fewer informative parameters within the model. Consequently, this discernment results in estimated models that are even sparser than the true model.

\section{Additional simulation results for $\hat{\bm{\beta}}$ (ASR-$\hat{\bm{\beta}}$)}
\label{app:Simulation Results beta}

\subsection{ASR-$\hat{\bm{\beta}}$: Choice of optimal penalty parameters by cross-validation}

\begin{figure}[h]
    \centering
    \begin{subfigure}[b]{0.9\textwidth} \centering
        \includegraphics[width=0.7\textwidth]{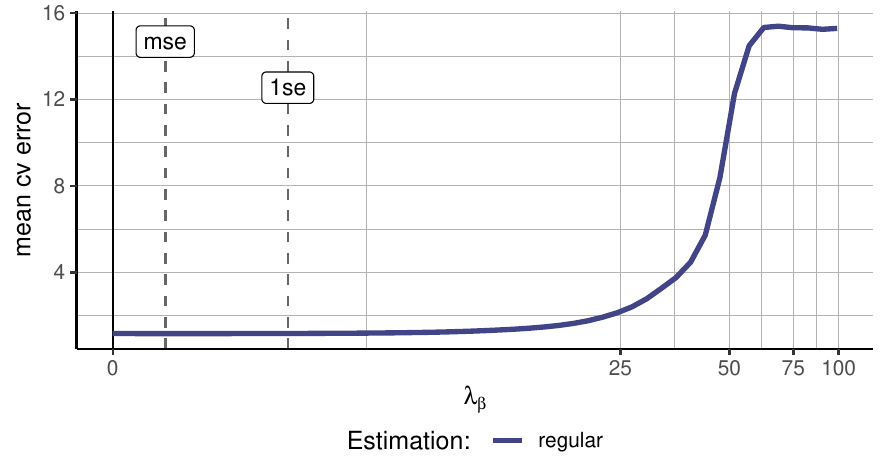}
        \caption{Mean error of out-of-sample prediction over cross-validation folds of $\bm{\beta}$ for different values of regularization parameter $\lambda_\beta$}
        \label{fig:lambda_b_cv} 
    \end{subfigure}\vspace{12pt}

    \begin{subfigure}[b]{0.9\textwidth} \centering
        \includegraphics[width=0.7\textwidth]{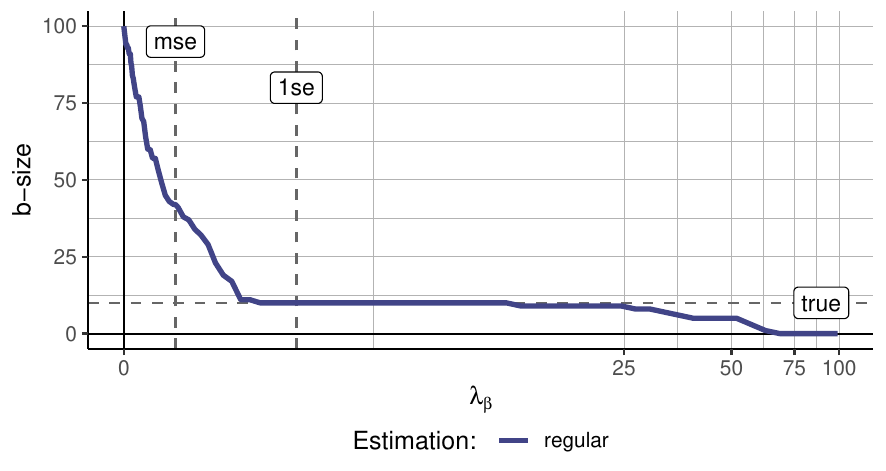}
        \caption{Number of non-zero elements in $ \bm{\beta}$ for different values of regularization parameter $\lambda_\beta$. The number of non-zero elements is equal for regular and post-regularization estimation}
        \label{fig:lambda_b_size}
    \end{subfigure}

    \caption{Choice of $\bm{\beta}$ by cross-validation \newline{\small Note: 'mse' and '1se' show determined $\lambda_\beta$ based on \textit{mse} and \textit{1se} optimality criteria. 'true' denotes the true value in the data.}}
    \label{fig:lambda_b}
\end{figure}

Figure \ref{fig:lambda_b} presents the outcomes of the cross-validation-based determination of the penalization parameter $\bm{\beta}$. The inherent convexity of the optimization problem is evident in Figure \ref{fig:lambda_b_size}, facilitating efficient numerical evaluation. The findings indicate that the utilization of cross-validation samples tends toward conservatism \citep{hastie2009elements}. Notably, the \textit{1se} criterion opts for a notably diminished model size while maintaining proximity to the optimal value of the objective function. The empirical findings illustrated in Figure \ref{fig:lambda_b_size} underscore that the model size determined by the \textit{1se} criterion aligns closely with the true number of non-zero parameters. In contrast, the penalization using the \textit{mse} optimality criterion appears to not shrink the model strongly enough.

\subsection{ASR-$\hat{\bm{\beta}}$: Model selection property}\label{app:beta_model_selection_property}

\begin{figure}[h]
    \centering
    \begin{subfigure}[b]{0.9\textwidth} \centering
        \includegraphics[width=0.7\textwidth]{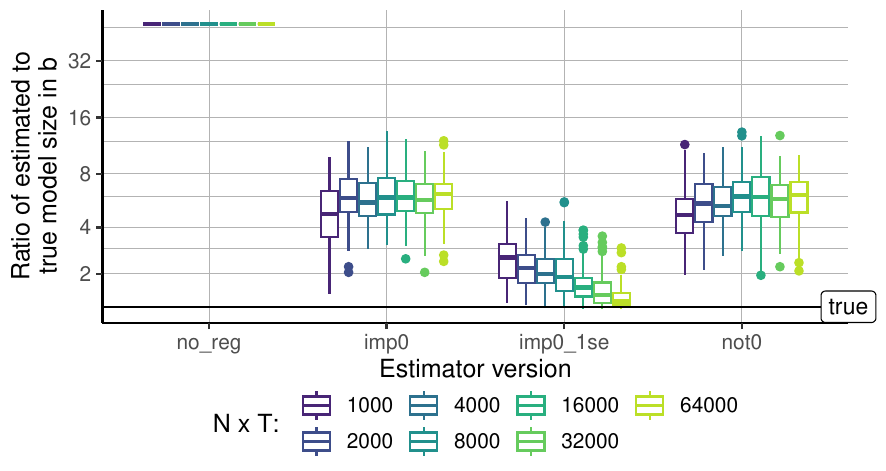}
    \caption{Ratio between the estimated and true model size in $\Hbf$.\newline{\small Note: $N=100$, $T \in \{10, 20, 40, 80, 160, 320, 640\}$. 700 simulations per sample size. The number of non-zero elements is equal for regular, rule-of-thumb correction, and post-regularization estimation. Transformed y-axis.}}
    \label{fig:sim_boxplot_sample_size_b_size_ratio} 
    \end{subfigure}\vspace{12pt}

    \begin{subfigure}[b]{0.9\textwidth} \centering
        \includegraphics[width=0.7\textwidth]{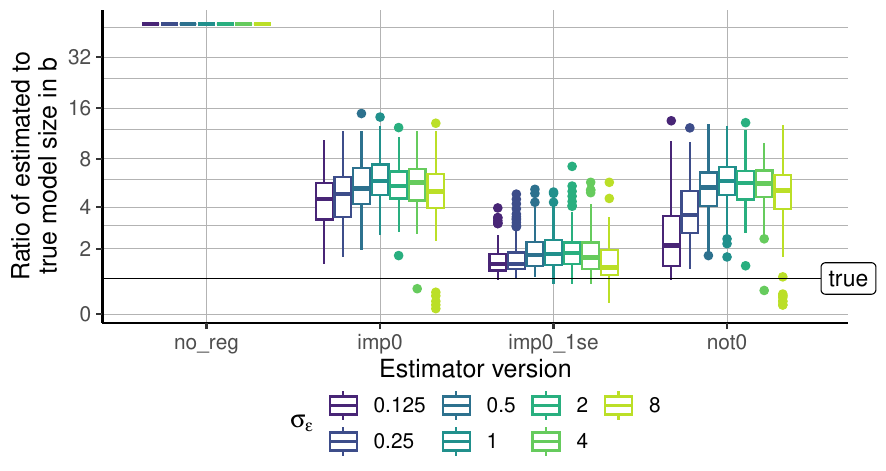}
    \caption{Ratio between the estimated and true model size in $\Hbf$.\newline{\small Note: Signal-to-noise ratio is given by $1/\sigma_\epsilon$. 1400 simulations per signal strength. $N=100$, $T=80$. The number of non-zero elements is equal for regular, rule-of-thumb correction, and post-regularization estimation. Transformed y-axis.}}
    \label{fig:sim_boxplot_signal_strength_b_size_ratio} 
    \end{subfigure}

    \caption{Ratio between the estimated and true model size in $\bm{\beta}$ for different (a) sample size and (b) signal strength.}
    \label{fig:sim_boxplot_b}
\end{figure}

We compare the dimensions of the obtained model, given by the count of non-zero coefficients in the $\bm{\beta}$, to the true size of the generated sample in order to evaluate the ability of the proposed estimator in model selection. Note that neither the post-regularization nor the rule-of-thumb correction estimates alter the model size and, accordingly, are omitted in the following results.\par

Figure~\ref{fig:sim_boxplot_sample_size_b_size_ratio} illustrates a pronounced reduction in the count of non-zero coefficients within the vector $\bm{\beta}$ when employing a regularized estimator. he baseline \textit{mse} optimal cross-validation yields a model that exceeds the true size by approximately sixfold. Contrastingly, employing the \textit{1se} criterion in cross-validation effectively contracts the non-zero coefficients in $\bm{\beta}$, aligning them closely with the correct model size. Notably, with an expanding sample size, the \textit{1se} estimator adeptly exploits the available data, leading to a convergence of the determined model size towards its true value.\par

Figure \ref{fig:sim_boxplot_signal_strength_b_size_ratio} delineates that the ascertained model size exhibits resilience as soon as the signal reaches minimal strengths. Nevertheless, as the signal-to-noise ratio is at a very small magnitude, all estimators tend towards discerning fewer parameters as informative. This propensity leads to the determination of a sparser model that aligns more closely with the true model size. However, it is important to underscore that this observed trend is unlikely to stem from enhanced model selection capabilities. Rather, it is probable that in scenarios characterized by weak signal strengths, the task of pinpointing a relevant parameter becomes inherently more challenging.

\end{document}